\documentclass[12pt]{article}
\usepackage[left=35mm,right=35mm,top=35mm,bottom=35mm]{geometry}
\usepackage{amsmath}
\usepackage{amsthm}
\usepackage{amssymb}
\usepackage{amsfonts}
\usepackage{mathrsfs}

\DeclareMathOperator*{\supp}{supp}

\DeclareMathOperator*{\dom}{Domain}

\makeatletter

\@addtoreset{equation}{section}
\makeatother
\newtheorem{The}{Theorem}[section]
\newtheorem{Ass}[The]{Assumption}
\newtheorem{Rem}[The]{Remark}
\newtheorem{Prop}[The]{Proposition}
\newtheorem{Cor}[The]{Corollary}
\newtheorem{Lem}[The]{Lemma}

\begin{document}
\begin{flushleft}
{ \Large \bf Minimal velocity bound for Schr\"odinger-type operator with fractional powers}
\end{flushleft}

\begin{flushleft}
{\large Atsuhide ISHIDA}\\
{Katsushika Division, Institute of Arts and Sciences, Tokyo University of Science, 6-3-1 Niijuku, Katsushika-ku, Tokyo 125-8585, Japan\\ 
Alfr\'ed R\'enyi Institute of Mathematics, Re\'altanoda utca 13-15, Budapest 1053, Hungary\\
Email: aishida@rs.tus.ac.jp
}
\end{flushleft}

\begin{abstract}
It is known in scattering theory that the minimal velocity bound plays a conclusive role in proving the asymptotic completeness of the wave operators. In this study, we prove the minimal velocity bound and other important estimates for the two-body Schr\"odinger-type operator with fractional powers. We assume that the pairwise potential functions belong to broad classes that include long-range decay and Coulomb-type local singularities. Our estimates are expected to be applied to prove the asymptotic completeness for the fractional Schr\"odinger-type operators in various (not only short-range but also long-range and $N$-body) situations.
\end{abstract}

\quad\textit{Keywords}: 
\begin{minipage}[t]{105mm}
Scattering theory, Fractional Schr\"odinger operator, Relativistic Schr\"odinger operator, Mourre estimate, Propagation estimate, Asymptotic completeness
\end{minipage}
\section{Introduction\label{introduction}}
For $s\geqslant0$, let us define the function
\begin{equation}
\Psi_\rho(s)=(s+1)^\rho-1\label{psi1}
\end{equation}
with $0<\rho\leqslant1$. The free dynamics that we consider in this paper are denoted by the symbol $\Psi_\rho$ of the Laplacian
\begin{equation}
\Psi_\rho\left(-\Delta\right)=\Psi_\rho\left(|D|^2\right)\label{free_hamiltonian}
\end{equation}
as a self-adjoint operator acting on $L^2(\mathbb{R}^n)$, where $D$ is the momentum operator $D=-{\rm i}\nabla=-\sqrt{-1}(\partial_1,\ldots,\partial_n)$ with $\partial_j=\partial_{x_j}$ for $1\leqslant j\leqslant n$. More precisely, $\Psi_\rho(|D|^2)$ is defined by the Fourier multiplier
\begin{equation}
\Psi_\rho\left(|D|^2\right)\phi(x)=\mathscr{F}^*\Psi_\rho\left(|\xi|^2\right)\mathscr{F}\phi(x)=\frac{1}{(2\pi)^n}\int\hspace{-2mm}\int_{\mathbb{R}^{2n}}e^{-{\rm i}(x-y)\cdot\xi}\Psi_\rho\left(|\xi|^2\right)\phi(y){\rm d}y{\rm d}\xi
\end{equation}
for $\phi\in H^{2\rho}(\mathbb{R}^n)$, which is the Sobolev space with order $2\rho$, where $\mathscr{F}$ and $\mathscr{F}^*$ respectively denote the Fourier transform on $L^2(\mathbb{R}^n)$ and its inverse. In particular, when $\rho=1$, $\Psi_1(|D|^2)$ is coincident with $-\Delta$ itself and when $\rho=1/2$, $\Psi_{1/2}(|D|^2)$ represents $\sqrt{-\Delta+1}-1$, which is, as is well known, the massive relativistic Schr\"odinger operator. $\Psi_\rho(|D|^2)$ is therefore the generalized Schr\"odinger-type operator in this sense. The full Hamiltonian $H_\rho$ is perturbated by the pairwise interaction potential $V=V(x)$ that is a multiplication operator of the function $V:\mathbb{R}^n\rightarrow\mathbb{R}$; i.e.,
\begin{equation}
H_\rho=\Psi_\rho\left(|D|^2\right)+V.\label{full_hamiltonian}
\end{equation}
In our study, we treat the general potentials that belong to classes as broadly as possible. To prove the main theorem, the minimal velocity bound in Theorem \ref{the1}, and other theorems in Section \ref{section_middle_velocity_bound} and \ref{section_minimal_velocity_bound}, we assume that the value of $V$ vanishes for $|x|\rightarrow\infty$, where $x=(x_1,\ldots,x_n)\in\mathbb{R}^n$. In contrast, to prove the maximal velocity bound in Theorem \ref{the2}, it is sufficient to assume the weaker conditions that guarantee only the self-adjointness of $H_\rho$. Further details are stated in Assumptions \ref{ass1} and \ref{ass2}.\par
In the following Assumption \ref{ass1}, the bracket of $x$ have the usual definition $\langle x\rangle=\sqrt{1+|x|^2}$. $A\lesssim B$ means that there exists a constant $C>0$ such that the inequality $A\leqslant CB$ holds. If emphasizing the dependence of $\alpha$ on $C=C_\alpha$, we write $A\lesssim_\alpha B$.

\begin{Ass}\label{ass1}
$V=V(x)$ is a real-valued-function and decomposes into the sum of three parts:
\begin{equation}
V=V_{\rm sing}+V_{\rm short}+V_{\rm long},\label{interaction1}
\end{equation}
where these real-valued functions $V_{\rm sing}$, $V_{\rm short}$ and $V_{\rm long}$ satisfy respective conditions.\par
$V_{\rm sing}=V_{\rm sing}(x)$ satisfies that, for $\gamma_{\rm sing}>1$, $\langle x\rangle^{\gamma_{\rm sing}}V_{\rm sing}$ belongs to $L^p(\mathbb{R}^n)$, where $p=2$ if $n<4\rho$ and $p>n/(2\rho)$ if $n\geqslant4\rho$.\par
$V_{\rm short}$=$V_{\rm short}(x)$ is bounded on $\mathbb{R}^n$ and has the spatial decay
\begin{equation}
|V_{\rm short}(x)|\lesssim\langle x\rangle^{-\gamma_{\rm short}},\label{short_range_decay}
\end{equation}
where $\gamma_{\rm short}>1$.\par
$V_{\rm long}$=$V_{\rm long}(x)$ belongs to $C^1(\mathbb{R}^n)$ and, for the multi-indices $\beta$ with $0\leqslant|\beta|\leqslant1$, has the spatial decay
\begin{equation}
|\partial_x^\beta V_{\rm long}(x)|\lesssim_\beta\langle x\rangle^{-\gamma_{\rm long}-|\beta|},\label{long_range_decay}
\end{equation}
where $\gamma_{\rm long}>0$.
\end{Ass}

\begin{Rem}\label{rem1}
Although $V_{\rm short}$ and $V_{\rm long}$ are bounded functions, $V_{\rm sing}$ is not always bounded. We therefore have to consider the self-adjointness of $H_\rho$. We provide the proof later in Proposition {\rm\ref{prop1}}. In the case where $0<\rho\leqslant1/4$, we cannot assume that $\langle x\rangle^{\gamma_{\rm sing}}V_{\rm sing}$ belongs to $L^2(\mathbb{R}^n)$ because the dimension $n$ has to satisfy $n<4\rho\leqslant1$. In the case where $3/4<\rho<1$, $V_{\rm sing}$ satisfies that $\langle x\rangle^{\gamma_{\rm sing}}V_{\rm sing}$ belongs to $L^p(\mathbb{R}^n)$, where p=2 if $n\leqslant3$ and $p>n/(2\rho)$ if $n\geqslant4$. These conditions are almost the same as the self-adjointness of the standard Schr\"odinger operator, $-\Delta$ and perturbational potentials; however, we cannot admit that $n/2\leqslant p\leqslant n/(2\rho)$ in our case even though $n\geqslant5$.
\end{Rem}

\begin{Rem}\label{rem2}
If $3/4<\rho\leqslant1$ and $n=3$, the Coulomb-type local singularity, for $\kappa\in\mathbb{R}$,
\begin{equation}
V_{\rm sing}(x)=\kappa|x|^{-1}F\left(|x|\leqslant1\right)\label{rem2_1}
\end{equation}
is admitted, where $F(\cdots)$ denotes the characteristic function of the set $\{\cdots\}$. Practically speaking,
\begin{equation}
\int_{|x|\leqslant1}|x|^{-2}{\rm d}x=\omega_n\int_0^1r^{-3+n}{\rm d}r\label{rem2_2}
\end{equation}
is bounded for $n=3$ and \eqref{rem2_1} belongs to $L^2(\mathbb{R}^3)$, where $\omega_n$ is the surface area of the $n$-dimensional unit sphere. If $1/2<\rho\leqslant3/4$, \eqref{rem2_1} does not belong to $L^2(\mathbb{R}^n)$ for $n=1$ and $2$ because \eqref{rem2_2} is the divergent integral. In this case, when choosing $p$ such that $n/(2\rho)<p<n$,
\begin{equation}
\int_{|x|\leqslant1}|x|^{-p}{\rm d}x=\omega_n\int_0^1r^{-p+n-1}{\rm d}r\label{rem2_3}
\end{equation}
is bounded. Therefore, when $1/2<\rho\leqslant3/4$, \eqref{rem2_1} belongs to $L^p(\mathbb{R}^n)$ and is admitted for all $n\geqslant3$. If $0<\rho\leqslant1/2$, we cannot admit \eqref{rem2_1} no matter what the dimension is. However,
\begin{equation}
V_{\rm sing}(x)=\kappa|x|^{-1+\epsilon}F\left(|x|\leqslant1\right)\label{rem2_4}
\end{equation}
with $\epsilon>0$ that satisfies $1-2\rho<\epsilon<1$ is admitted for all $n\geqslant1$. This is because we can take $p$ that satisfies $n/(2\rho)<p<n/(1-\epsilon)$, and
\begin{equation}
\int_{|x|\leqslant1}|x|^{(-1+\epsilon)p}{\rm d}x=\omega_n\int_0^1r^{(-1+\epsilon)p+n-1}{\rm d}r\label{rem2_5}
\end{equation}
is bounded for $(-1+\epsilon)p+n-1 >-1$ which is equivalent to $p<n/(1-\epsilon)$. This implies that \eqref{rem2_4} belongs to $L^p(\mathbb{R}^n)$.
\end{Rem}

The main result of this paper is the propagation estimate that has the following integral-form. In scattering theory, we often refer to this estimate as the minimal velocity bound. We denote the pure point spectrum of $H_\rho$ by $\sigma_{\rm pp}(H_\rho)$. 

\begin{The}\label{the1}
\textbf{Minimal velocity bound.} Let $f\in C_0^\infty((0,\infty))$ satisfy $\supp f\cap\sigma_{\rm pp}(H_\rho)=\emptyset$ and $\theta_0>0$ be sufficiently small. Then, the inequality
\begin{equation}
\int_1^\infty\left\|F\left(\frac{|x|}{2t}\leqslant\theta_0\right)f(H_\rho)e^{-{\rm i}tH_\rho}\phi\right\|_{L^2(\mathbb{R}^n)}^2\frac{{\rm d}t}{t}\lesssim\|\phi\|_{L^2(\mathbb{R}^n)}^2\label{minimal_velocity_bound}
\end{equation}
holds for $\phi\in L^2(\mathbb{R}^n)$.
\end{The}

This propagation estimate \eqref{minimal_velocity_bound} is powerful and if $V$ has the short-range parts only, the asymptotic completeness of the wave operators can be obtained immediately. With regard to the long-range case, if we construct some type of modification of the wave operators, Theorem \ref{the1} can also be applied to the proof of the asymptotic completeness of the modified wave operators. Moreover, it is known that the propagation estimates of the integral-form are available to the $N$-body case. As we state below, although there are some results of the minimal velocity bound with the integral-form for the standard Schr\"odinger time evolution, the case for the Schr\"odinger-type operator with the general fractional powers has not been discussed up until our study. Of course, in the fractional powers, the relativistic quantum case $\rho=1/2$ is most important physically. However, the both cases where $0<\rho<1/2$ and $1/2<\rho<1$ are of mathematical interest and challenging. For instance, in the case where $1/2<\rho<1$, $\Psi'_\rho(|D|^2)\langle D\rangle$ is not bounded and we have to make full use of the energy cut off. When $0<\rho<1/2$, $\Psi'_\rho(|D|^2)\langle D\rangle$ is bounded, while $\langle D\rangle\langle H_\rho\rangle^{-1}$ is not bounded and this difficulty affects parts of our discussions.\par
In section \ref{section_minimal_velocity_bound}, we prove the Mourre estimate in Theorem \ref{the4}. In our proof of Theorem \ref{the1}, the Mourre inequality also fulfills a crucial role. In Mourre theory, it is important to find a conjugate operator. We employ the choice $A_\rho$ (see \eqref{conjugate2}) and prove the isolatedness and finite multiplicity of $\sigma_{\rm pp}(H_\rho)\setminus\{0\}$ in Corollary \ref{cor1} using the Mourre inequality.\par
It seems that the minimal velocity bound with the integral-form was first obtained by Sigal and Soffer \cite[Theorem 4.2]{SiSo} for the long-range and $N$-body Schr\"odinger operator. We currently refer to the works of Derezi\'nski and G\'erard \cite[Propositions 4.4.7 and 6.6.8]{DeGe} and Isozaki \cite[Theorems 2.38 and 3.11]{Iso}, which explain in detail the method of reaching the minimal velocity bound for the standard Schr\"odinger operators in the cases of two- to $N$-body. In the same manner as for the standard Schr\"odinger case, in proving Theorem \ref{the1}, we need the maximal velocity bound in Theorem \ref{the2} and the middle velocity bound in Theorem \ref{the3}. The maximal velocity bound with the integral-form was first proved by \cite[Theorem 4.3]{SiSo}. The middle velocity bound with the integral-form was first proved by Graf \cite[Theorem 4.3]{Gr} for the short-range $N$-body Schr\"odinger operator. Meanwhile, the minimal velocity bound with pointwise-form initiated by Skibsted \cite{Sk} and G\'erard \cite{Ge} is also an important estimate with which to prove the asymptotic completeness. The pointwise-form of the conjugate operator was developed by Hunziker, Sigal and Soffer \cite{HuSiSo} and Richard \cite{Ri} in the abstract settings.\par
Scattering theory for the Schr\"odinger-type operator with fractional powers has been studied. Gire \cite{Gi} considered general functions of the Laplacian that included the relativistic Schr\"odinger operator and discussed the asymptotic completeness for the short-range potentials by investigating the semigroup differences. Kitada \cite{Ki1,Ki2} constructed long-range scattering theory for the fractional Laplacian $(-\Delta)^\rho$ with $1/2\leqslant\rho\leqslant1$ adopting the Enss method and smooth perturbation theory. Ishida \cite{Is} studied inverse scattering for $(-\Delta)^\rho$ with $1/2<\rho\leqslant1$. It is noteworthy that although it was only the case of the massless relativistic Schr\"odinger operator $\sqrt{-\Delta}$, Soffer \cite{So} obtained the integral-form minimal velocity bound by using its pointwise maximal velocity bound. Ishida and Wada \cite{IsWa} considered non-local Schr\"odinger operators that included the Bernstein functions of the Laplacian, and they decided the threshold between short- and long-range decay conditions of the potential functions by providing a counter-example such that the wave operators did not exist.\par
At the end of this section, we prove the self-adjointness of $H_\rho$. By virtue of the Kato--Rellich theorem (Reed and Simon \cite[Theorem X.12]{ReSi}) and following Proposition \ref{prop1}, if $V$ satisfies Assumption \ref{ass1}, then $H_\rho=\Psi_\rho\left(|D|^2\right)+V$ is essentially self-adjoint with the core $C_0^\infty(\mathbb{R}^n)$. The original idea of this proof for the standard Schr\"odinger operator $-\Delta+V$ is explained in \cite[Theorem X.20]{ReSi} and \cite[Lemma 1.9]{Iso}.

\begin{Prop}\label{prop1}
Suppose the real-valued function $\hat{V}_{\rm sing}=\hat{V}_{\rm sing}(x)$ satisfies that $\hat{V}_{\rm sing}$ belongs to $L^p(\mathbb{R}^n)$, where $p=2$ if $n<4\rho$ and $p>n/(2\rho)$ if $n\geqslant4\rho$. Then, for any $\epsilon>0$ and $\phi\in C_0^\infty(\mathbb{R}^n)$, there exists a constant $C_\epsilon>0$ such that
\begin{equation}
\bigl\|\hat{V}_{\rm sing}\phi\bigr\|_{L^2(\mathbb{R}^n)}\leqslant\epsilon\left\|\Psi_\rho\left(|D|^2\right)\phi\right\|_{L^2(\mathbb{R}^n)}+C_\epsilon\|\phi\|_{L^2(\mathbb{R}^n)}\label{relatively_bound1}
\end{equation}
holds.
\end{Prop}

\begin{proof}[Proof of Proposition \ref{prop1}]
Let $0<\delta<1/2$. We note that
\begin{gather}
\int_{\mathbb{R}^n}\frac{{\rm d}\xi}{\left\{1+\delta\Psi_\rho\left(|\xi|^2\right)\right\}^p}=\omega_n\delta^{-n/(2\rho)}\int_0^\infty\frac{\eta^{n-1}{\rm d}\eta}{\left\{1-\delta+\left(\delta^{1/\rho}+\eta^2\right)^\rho\right\}^p}\nonumber\\
\leqslant\omega_n\delta^{-n/(2\rho)}\left\{2^p+\int_1^\infty \eta^{n-1-2\rho p}{\rm d}\eta\right\}
\end{gather}
with a changing variable $\eta=\delta^{1/(2\rho)}|\xi|$, and that
\begin{equation}
\left\|\left\{1+\delta\Psi_\rho\left(|\xi|^2\right)\right\}^{-1}\right\|_{L^p(\mathbb{R}_\xi^n)}\lesssim\delta^{-n/(2\rho p)}\label{prop1_1}
\end{equation}
because $n-1-2\rho p<-1$. If $n<4\rho$, we express $\phi\in C_0^\infty(\mathbb{R}^n)$ by
\begin{equation}
\phi(x)=\frac{1}{(2\pi)^{n/2}}\int_{\mathbb{R}^n}e^{{\rm i}x\cdot\xi}\left\{1+\delta\Psi_\rho\left(|\xi|^2\right)\right\}^{-1}\left\{1+\delta\Psi_\rho\left(|\xi|^2\right)\right\}\mathscr{F}\phi(\xi){\rm d}\xi
\end{equation}
and there exists $C>0$ such that
\begin{gather}
|\phi(x)|\lesssim\left\|\left\{1+\delta\Psi_\rho\left(|\xi|^2\right)\right\}^{-1}\right\|_{L^2(\mathbb{R}_\xi^n)}\left\|\left\{1+\delta\Psi_\rho\left(|\xi|^2\right)\right\}\mathscr{F}\phi\right\|_{L^2(\mathbb{R}_\xi^n)}\nonumber\\
\leqslant C\left(\delta^{1-n/(4\rho)}\left\|\Psi_\rho\left(|D|^2\right)\phi\right\|_{L^2}+\delta^{-n/(4\rho)}\|\phi\|_{L^2}\right),\label{prop1_2}
\end{gather}
using the Schwarz inequality and \eqref{prop1_1} for $p=2$. If making $\delta$ small such that $C\delta^{1-n/(4\rho)}\|\hat{V}_{\rm sing}\|_{L^2}\leqslant\epsilon$, then \eqref{prop1_2} and
\begin{equation}
\bigl\|\hat{V}_{\rm sing}\phi\bigr\|_{L^2}\leqslant\bigl\|\hat{V}_{\rm sing}\bigr\|_{L^2}\sup_{x\in\mathbb{R}^n}|\phi(x)|
\end{equation}
imply \eqref{relatively_bound1}. We next assume that $n\geqslant4\rho$ and $p>n/(2\rho)$. For $q_1=2p/(p-2)$, by the H\"older inequality,
\begin{equation}
\bigl\|\hat{V}_{\rm sing}\phi\bigr\|_{L^2}\leqslant\bigl\|\hat{V}_{\rm sing}\bigr\|_{L^p}\|\phi\|_{L^{q_1}}.
\end{equation}
holds. For $q_2=q_1/(q_1-1)=2p/(p+2)$, by the Hausdorff--Young inequality \cite[Theorem IX.8]{ReSi}, we have
\begin{equation}
\|\phi\|_{L^{q_1}}\leqslant(2\pi)^{n(1/2-1/q_2)}\left\|\mathscr{F}\phi\right\|_{L^{q_2}}
\end{equation}
noting that $q_1>2$ and $1<q_2<2$. Using the H\"older inequality and \eqref{prop1_1} again, we have
\begin{gather}
\left\|\mathscr{F}\phi\right\|_{L^{q_2}}\leqslant\left\|\left\{1+\delta\Psi_\rho\left(|\xi|^2\right)\right\}^{-1}\right\|_{L^p(\mathbb{R}_\xi^n)}\left\|\left\{1+\delta\Psi_\rho\left(|\xi|^2\right)\right\}\mathscr{F}\phi\right\|_{L^2(\mathbb{R}_\xi^n)}\nonumber\\
\lesssim\delta^{1-n/(2\rho p)}\left\|\Psi_\rho\left(|D|^2\right)\phi\right\|_{L^2}+\delta^{-n/(2\rho p)}\|\phi\|_{L^2}.
\end{gather}
This completes the proof.
\end{proof} 

\section{Maximal velocity bound}\label{section_maximal_velocity_bound}
In this section, we prove the propagation estimate for the high-velocity region in Theorem \ref{the2}, which is needed for the proof of Theorem \ref{the3} in the next section. We often refer to this estimate as the maximal velocity bound. If we prove Theorem \ref{the2} only, the bounded parts of the potential function $V_{\rm short}+V_{\rm long}$ do not necessarily disappear for $|x|\rightarrow\infty$, and the singular part $V_{\rm sing}$ can decay far more slowly. Throughout this section, instead of Assumption \ref{ass1}, we assume the following Assumption \ref{ass2}.

\begin{Ass}\label{ass2}
$V=V(x)$ is a real-valued function and decomposes into the sum of two parts:
\begin{equation}
V=\hat{V}_{\rm sing}+V_{\rm bdd},\label{interaction2}
\end{equation}
where $\hat{V}_{\rm sing}=\hat{V}_{\rm sing}(x)$ satisfies the conditions in Proposition \ref{prop1} while $V_{\rm bdd}=V_{\rm bdd}(x)$ belongs to $L^\infty(\mathbb{R}^n)$.
\end{Ass}

\begin{Rem}\label{rem3}
In the case where $\rho=1$ {\rm(}i.e., the standard Schr\"odinger operator case{\rm )}, $V_{\rm bdd}$ in \eqref{interaction2} can be replaced with $\check{V}_{\rm sing}=\check{V}_{\rm sing}(x)$ that belongs to $L_{\rm loc}^2(\mathbb{R}^n)$ and satisfies
\begin{equation}
\check{V}_{\rm sing}(x)\gtrsim-\langle x\rangle^2
\end{equation}
by applying the Kato distributional inequality {\rm\cite[Theorems X.27]{ReSi}} and Faris-Lavine theorem {\rm\cite[Theorems X.38]{ReSi}}. This means that the potential function $V$ can be allowed to grow in $x$ to prove Theorem {\rm\ref{maximal_velocity_bound}} only. We note that $\dom(-\Delta+V)$ does not always coincide with $H^2(\mathbb{R}^n)$ in this case.
\end{Rem}

Under Assumption \ref{ass2}, $H_\rho=\Psi_\rho(|D|^2)+V$ is self-adjoint by Proposition \ref{prop1}. We here note again that if $V$ satisfies Assumption \ref{ass1}, $V$ also satisfies Assumption \ref{ass2}. The maximal velocity bound is stated as the following theorem. The corresponding propagation estimate for the standard two-body Schr\"odinger operator is detailed in \cite[Proposition 4.2.1]{DeGe} and \cite[Theorem 2.31]{Iso}.

\begin{The}\label{the2}
\textbf{Maximal velocity bound.} Take $f\in C_0^\infty(\mathbb{R})$ arbitrarily. There exists $\Theta>0$ such that, for any $\theta>\Theta$ and $\phi\in L^2(\mathbb{R}^n)$,
\begin{equation}
\int_1^\infty\left\|F\left(\Theta\leqslant\frac{|x|}{2t}\leqslant\theta\right)f(H_\rho)e^{-{\rm i}tH_\rho}\phi\right\|_{L^2(\mathbb{R}^n)}^2\frac{{\rm d}t}{t}\lesssim\|\phi\|_{L^2(\mathbb{R}^n)}^2\label{maximal_velocity_bound}
\end{equation}
holds.
\end{The}

We provide preparations in advance of the proof of Theorem \ref{the2}. To analyze $\Psi_\rho(|D|^2)$ as a function of the Laplacian, we make efficient use of the almost analytic extension and commutator expansions. We thus extend the domain of $\Psi_\rho(s)$ to a full real axis and employ the function $\Psi_\rho\in\ C^\infty(\mathbb{R})$ such that
\begin{equation}
\Psi_\rho(s)=
\begin{cases}
\ (s+1)^\rho-1 &\quad \mbox{if}\ s\geqslant0,\\
\ 0&\quad \mbox{if}\ s\leqslant-1
\end{cases}
\end{equation}
for $0<\rho<1$. This $\Psi_\rho$ satisfies, for all $k\in\mathbb{N}\cup\{0\}$,
\begin{equation}
\left|\frac{{\rm d}^k}{{\rm d}s^k}\Psi_\rho(s)\right|\lesssim_k\langle s\rangle^{\rho-k}
\end{equation}
on $\mathbb{R}$. We therefore find a function $\tilde{\Psi}_\rho\in C^\infty(\mathbb{C})$, called an almost analytic extension of $\Psi_\rho$ \cite[Propositions C.2.1 and C.2.2]{DeGe}; i.e., $\tilde{\Psi}_\rho$ with
\begin{equation}
\supp\tilde{\Psi}_\rho\subset\left\{z\in\mathbb{C}\bigm||{\rm Im}z|\lesssim\langle{\rm Re}z\rangle\right\}
\end{equation}
satisfies that $\tilde{\Psi}_\rho(s)=\Psi_\rho(s)$ for $s\in\mathbb{R}$ and that
\begin{equation}
\bigl|\bar{\partial_z}\tilde{\Psi}_\rho(z)\bigr|\lesssim_N|{\rm Im}z|^N\langle z\rangle^{\rho-1-N}\label{almost_analytic_extension}
\end{equation}
for $N\in\mathbb{N}$, where $\bar{\partial_z}=(\partial_{{\rm Re}z}+{\rm i}\partial_{{\rm Im}z})/2$. One of the most effective applications of the almost analytic extension is the Helffer--Sj\"ostrand formula originated by \cite[Proposition 7.2]{HeSj} (see also \cite[Theorem 1.17]{Iso}). Unfortunately, we can not apply this formula to $\Psi_\rho(|D|^2)$ directly because $\rho>0$. However, when $0<\rho<1$, we can consider commutator expansions with a function of $x$ by applying the Helffer--Sj\"ostrand formula to $\Psi_\rho/(1+s)$ instead of $\Psi_\rho$. The more general settings of the commutator expansions are referred to \cite[Lemma C.3.1]{DeGe} and \cite[Definition 4.11]{Iso}.

\begin{Lem}\label{lem1}
Suppose $0<\rho<1$ and put $\Phi_\rho(s)=\Psi_\rho(s)/(1+s)$. For a smooth function $\chi=\chi(x)$ such that its all derivatives are bounded, the commutator $[\Psi_\rho(|D|^2),\chi]$ has the expansions
\begin{gather}
\left[\Psi_\rho\left(|D|^2\right),\chi\right]=\left[|D|^2,\chi\right]\Psi'_\rho\left(|D|^2\right)\nonumber\\
+\ \frac{1}{2\pi{\rm i}}\int_\mathbb{C}\bigl(\bar{\partial_z}\tilde{\Psi}_\rho\bigr)(z)\left(z-|D|^2\right)^{-1}\left[|D|^2,\left[|D|^2,\chi\right]\right]\left(z-|D|^2\right)^{-2}{\rm d}z\wedge{\rm d}\bar{z},\label{commutator_left}
\end{gather}
and
\begin{gather}
\left[\Psi_\rho\left(|D|^2\right),\chi\right]=\Psi'_\rho\left(|D|^2\right)\left[|D|^2,\chi\right]\nonumber\\
-\ \frac{1}{2\pi{\rm i}}\int_\mathbb{C}\bigl(\bar{\partial_z}\tilde{\Psi}_\rho\bigr)(z)\left(z-|D|^2\right)^{-2}\left[|D|^2,\left[|D|^2,\chi\right]\right]\left(z-|D|^2\right)^{-1}dz\wedge d\bar{z},\label{commutator_right}
\end{gather}
where ${\rm d}z\wedge{\rm d}\bar{z}=-2{\rm i}{\rm d}{{\rm Re}z}\wedge{\rm d}{{\rm Im}z}$ is the two-dimensional Lebesgue measure and $\Psi'_\rho$ denotes ${\rm d}\Psi_\rho/{\rm d}s$.
\end{Lem}

\begin{Rem}\label{rem4}
Right-hand sides of \eqref{commutator_left} and \eqref{commutator_right} are operators on $H^{2\rho}(\mathbb{R}^n)$ because the integral terms are bounded by \eqref{lem1_1} and
\begin{equation}
\left[|D|^2,\chi\right]=-{\rm i}D\cdot\nabla\chi-{\rm i}\nabla\chi\cdot D=-2{\rm i}D\cdot\nabla\chi+\Delta\chi=-2{\rm i}\nabla\chi\cdot D-\Delta\chi
\end{equation}
holds on $H^2(\mathbb{R}^n)$.
\end{Rem}

\begin{proof}[Proof of Lemma \ref{lem1}]
We prove the formula \eqref{commutator_left} only. We first note that
\begin{equation}
\int_\mathbb{C}\bigl|\bar{\partial_z}\tilde{\Psi}_\rho(z)\bigr|\left\|\left(z-|D|^2\right)^{-1}\left[|D|^2,\left[|D|^2,\chi\right]\right]\left(z-|D|^2\right)^{-2}\right\|\left|{\rm d}z\wedge{\rm d}\bar{z}\right|<\infty,\label{lem1_1}
\end{equation}
where we denote the operator norm on $L^2(\mathbb{R}^n)$ by $\|\cdot\|$. This is seen as follows. By the basic inequality
\begin{equation}
\sup_{\lambda\in\mathbb{R}}\frac{\langle \lambda\rangle^{q_1}}{|z-\lambda|^{q_2}}\lesssim_{q_1q_2}\frac{\langle z\rangle^{q_1}}{|{\rm Im} z|^{q_2}}\label{q1q2_inequality}
\end{equation}
for $q_2>0$ and $0\leqslant q_1\leqslant q_2$, we have
\begin{equation}
\left\|\left(z-|D|^2\right)^{-1}\left[|D|^2,\left[|D|^2,\chi\right]\right]\left(z-|D|^2\right)^{-2}\right\|\lesssim|{\rm Im}z|^{-3}\langle z\rangle
\end{equation}
for $z\in\mathbb{C}\setminus\mathbb{R}$. Inequality \eqref{q1q2_inequality} will be used often in our proof. Therefore, the left-hand side of \eqref{lem1_1} is bounded because, for $\rho<1$,
\begin{equation}
\int_\mathbb{C}\langle z\rangle^{\rho-3}\left|dz\wedge d\bar{z}\right|<\infty
\end{equation}
by \eqref{almost_analytic_extension} with $N=3$. We now will prove \eqref{commutator_left}. Because
\begin{equation}
\left|\frac{{\rm d}^k}{{\rm d}s^k}\Phi_\rho(s)\right|\lesssim_k\langle s\rangle^{\rho-1-k}\label{lem1_2}
\end{equation}
holds for $k\in\mathbb{N}\cup\{0\}$, an almost analytic extension $\tilde{\Phi}_\rho\in C^\infty(\mathbb{C})$ has the estimate
\begin{equation}
\bigr|\bar{\partial_z}\tilde{\Phi}_\rho(z)\bigl|\lesssim_N|{\rm Im}z|^N\langle z\rangle^{\rho-2-N}\label{lem1_3}
\end{equation}
for any $N\in\mathbb{N}$. According to the Helffer--Sj\"ostrand formula, $\Phi_\rho(|D|^2)$ is expressed as
\begin{equation}
\Phi_\rho\left(|D|^2\right)=\frac{1}{2\pi{\rm i}}\int_\mathbb{C}\bigl(\bar{\partial_z}\tilde{\Phi}_\rho\bigr)(z)\left(z-|D|^2\right)^{-1}{\rm d}z\wedge{\rm d}\bar{z}.
\end{equation}
We therefore compute
\begin{gather}
\left[\Phi_\rho\left(|D|^2\right),\chi\right]=\frac{1}{2\pi{\rm i}}\int_\mathbb{C}\bigl(\bar{\partial_z}\tilde{\Phi}_\rho\bigr)(z)\left(z-|D|^2\right)^{-1}\left[|D|^2,\chi\right]\left(z-|D|^2\right)^{-1}{\rm d}z\wedge{\rm d}\bar{z}\nonumber\\
=\left[|D|^2,\chi\right]\Phi'_\rho\left(|D|^2\right)\nonumber\\
+\ \frac{1}{2\pi{\rm i}}\int_\mathbb{C}\bigl(\bar{\partial_z}\tilde{\Phi}_\rho\bigr)(z)\left(z-|D|^2\right)^{-1}\left[|D|^2,\left[|D|^2,\chi\right]\right]\left(z-|D|^2\right)^{-2}{\rm d}z\wedge{\rm d}\bar{z}.\label{lem1_4}
\end{gather}
Incidentally, from the definition of $\Phi_\rho$,
\begin{equation}
\left[\Phi_\rho\left(|D|^2\right),\chi\right]=\left[\Psi_\rho\left(|D|^2\right),\chi\right]\langle D\rangle^{-2}-\Phi_\rho\left(|D|^2\right)\left[|D|^2,\chi\right]\langle D\rangle^{-2}\label{lem1_5}
\end{equation}
and
\begin{equation}
\Phi'_\rho\left(|D|^2\right)=\Psi'_\rho\left(|D|^2\right)\langle D\rangle^{-2}-\Phi_\rho\left(|D|^2\right)\langle D\rangle^{-2}\label{lem1_6}
\end{equation}
hold. Combining \eqref{lem1_4}, \eqref{lem1_5}, and \eqref{lem1_6}, we have
\begin{gather}
\left[\Psi_\rho\left(|D|^2\right),\chi\right]\langle D\rangle^{-2}=\left[|D|^2,\chi\right]\Psi'_\rho\left(|D|^2\right)\langle D\rangle^{-2}+\left[\Phi_\rho\left(|D|^2\right),\left[|D|^2,\chi\right]\right]\langle D\rangle^{-2}\nonumber\\
+\ \frac{1}{2\pi{\rm i}}\int_\mathbb{C}\bigl(\bar{\partial_z}\tilde{\Phi}_\rho\bigr)(z)\left(z-|D|^2\right)^{-1}\left[|D|^2,\left[|D|^2,\chi\right]\right]\left(z-|D|^2\right)^{-2}{\rm d}z\wedge{\rm d}\bar{z}.\label{lem1_7}
\end{gather}
This equation implies
\begin{gather}
\left[\Psi_\rho\left(|D|^2\right),\chi\right]=\left[|D|^2,\chi\right]\Psi'_\rho\left(|D|^2\right)\nonumber\\
+\ \frac{1}{2\pi{\rm i}}\int_\mathbb{C}\bigl(\bar{\partial_z}\tilde{\Phi}_\rho\bigr)(z)(1+z)\left(z-|D|^2\right)^{-1}\left[|D|^2,\left[|D|^2,\chi\right]\right]\left(z-|D|^2\right)^{-2}{\rm d}z\wedge{\rm d}\bar{z},\label{lem1_8}
\end{gather}
noting that
\begin{gather}
\left[\Phi_\rho\left(|D|^2\right),\left[|D|^2,\chi\right]\right]\nonumber\\
=\frac{1}{2\pi{\rm i}}\int_\mathbb{C}\bigl(\bar{\partial_z}\tilde{\Phi}_\rho\bigr)(z)\left(z-|D|^2\right)^{-1}\left[|D|^2,\left[|D|^2,\chi\right]\right]\left(z-|D|^2\right)^{-1}{\rm d}z\wedge{\rm d}\bar{z}
\end{gather}
by the Helffer--Sj\"ostrand formula and that $\langle D\rangle^2=-(z-|D|^2)+1+z$. Because $\tilde{\Phi}_\rho(z)(1+z)$ corresponds with one of the almost analytic extensions of $\Psi_\rho$, we have
\begin{equation}
\bigl(\bar{\partial_z}\tilde{\Phi}_\rho\bigr)(z)(1+z)=\bar{\partial_z}\left\{\tilde{\Phi}_\rho(z)(1+z)\right\}=\bigl(\bar{\partial_z}\tilde{\Psi}_\rho\bigr)(z).\label{lem1_9}
\end{equation}
\eqref{lem1_8} and \eqref{lem1_9} imply \eqref{commutator_left}.
\end{proof}

We will use the following notations frequently. The Heisenberg derivative of a time-dependent operator $P(t)$ associated with an operator $Q$ is
\begin{equation}
\mathbb{D}_QP(t)=\frac{{\rm d}}{{\rm d}t}P(t)+{\rm i}\left[Q,P(t) \right].
\end{equation}
If $P$ is time-independent, $\mathbb{D}_QP$ is ${\rm i}[Q,P]$. $P(t)=\mathcal{O}(t^\nu)$ means that $P(t)$ is the bounded operator and that $\|P(t)\|\lesssim t^\nu$ for $\nu\in\mathbb{R}$. The Hermitian conjugate ${\rm hc}$ is defined by $Q+{\rm hc}=Q+Q^*$, where $Q^*$ is the formal adjoint of $Q$.

\begin{proof}[Proof of Theorem \ref{the2}]
Let $\chi\in C_0^\infty(\mathbb{R})$ satisfy that $\chi(s)=1$ if $\Theta/2\leqslant s\leqslant 2\theta$ and $\chi(s)=0$ if $s\leqslant\Theta/3$ for $0<\Theta<\theta$, where the size of $\Theta$ is to be determined below. Put $X(s)=\int^s_{-\infty}\chi(\tau)^2{\rm d}\tau$ and
\begin{equation}
\mathscr{L}(t)=f(H_\rho)X\left(\frac{|x|}{2t}\right)f(H_\rho),
\end{equation}
according to \cite[Proposition 4.2.1]{DeGe} and \cite[Theorem 2.31]{Iso}. Clearly, $\mathscr{L}(t)=\mathcal{O}(1)$. We first give the proof for the case where $\rho<1$. Using \eqref{commutator_right}, we compute
\begin{gather}
{\rm i}\left[\Psi_\rho\left(|D|^2\right),X\left(\frac{|x|}{2t}\right) \right]=\Psi'_\rho\left(|D|^2\right)\left\{\frac{1}{2t}D\cdot\frac{x}{|x|}\chi\left(\frac{|x|}{2t}\right)^2+{\rm hc}\right\}+\mathcal{O}\left(t^{-2}\right)\nonumber\\
=\frac{1}{2t}\Psi'_\rho\left(|D|^2\right)D\cdot\frac{x}{|x|}\chi\left(\frac{|x|}{2t}\right)^2+{\rm hc}+\mathcal{O}\left(t^{-2}\right).\label{the2_1}
\end{gather}
We here adopted the estimate
\begin{equation}
\left[\Psi'_\rho\left(|D|^2\right),\chi\left(\frac{|x|}{2t}\right)^2\frac{x}{|x|}\cdot D\right]=\mathcal{O}\left(t^{-1}\right),\label{the2_2}
\end{equation}
using the Helffer--Sj\"ostrand formula directly with
\begin{equation}
\left\|\left(z-|D|^2\right)^{-1}\left[|D|^2,\chi\left(\frac{|x|}{2t}\right)^2\frac{x}{|x|}\cdot D\right]\left(z-|D|^2\right)^{-1}\right\|\lesssim t^{-1}|{\rm Im}z|^{-2}\langle z\rangle
\end{equation}
and $|\bar{\partial}_z\tilde{\Psi}'_\rho(z)|\lesssim|{\rm Im}z|^2\langle z\rangle^{\rho-4}$. Therefore, from \eqref{the2_1}, we have
\begin{gather}
\mathbb{D}_{\Psi_\rho(|D|^2)}X\left(\frac{|x|}{2t}\right)=-\frac{|x|}{2t^2}\chi\left(\frac{|x|}{2t}\right)^2\nonumber\\
+\ \frac{1}{2t}\Psi'_\rho\left(|D|^2\right)D\cdot\frac{x}{|x|}\chi\left(\frac{|x|}{2t}\right)^2+{\rm hc}+\mathcal{O}\left(t^{-2}\right).\label{heisenberg1}
\end{gather}
We take $g\in C_0^\infty(\mathbb{R})$ such that $f=fg$ and compute
\begin{gather}
f(H_\rho)\Psi'_\rho\left(|D|^2\right)D\cdot\frac{x}{|x|}\chi\left(\frac{|x|}{2t}\right)^2f(H_\rho)\nonumber\\
=f(H_\rho)\chi\left(\frac{|x|}{2t}\right)g(H_\rho)\Psi'_\rho\left(|D|^2\right)D\cdot\frac{x}{|x|}\chi\left(\frac{|x|}{2t}\right)f(H_\rho)+I_1(t)+I_2(t)\label{the2_3}.
\end{gather}
We defined $I_1$ and $I_2$ in \eqref{the2_3} by
\begin{align}
I_1(t)&=f(H_\rho)\sum_{j=1}^n\left[\Psi'_\rho\left(|D|^2\right)D_j,\chi\left(\frac{|x|}{2t}\right)\right]\frac{x_j}{|x|}\chi\left(\frac{|x|}{2t}\right)f(H_\rho),\\
I_2(t)&=f(H_\rho)\left[g(H_\rho),\chi\left(\frac{|x|}{2t}\right)\right]\Psi'_\rho\left(|D|^2\right)D\cdot\frac{x}{|x|}\chi\left(\frac{|x|}{2t}\right)f(H_\rho),
\end{align}
where $D_j$ is the $j$th component of $D$. Making the same computation as \eqref{the2_2} yields
\begin{gather}
\left[\Psi'_\rho\left(|D|^2\right)D_j,\chi\left(\frac{|x|}{2t}\right)\right]\nonumber\\
=\Psi'_\rho\left(|D|^2\right)\mathcal{O}\left(t^{-1}\right)+\left[\Psi'_\rho\left(|D|^2\right),\chi\left(\frac{|x|}{2t}\right)\right]D_j=\mathcal{O}\left(t^{-1}\right)\label{the2_4}
\end{gather}
for $1\leqslant j\leqslant n$ and $I_1(t)=\mathcal{O}(t^{-1})$ holds. We here note that $\langle\Psi_\rho(|D|^2)\rangle\langle H_\rho\rangle^{-1}$ is bounded by virtue of Proposition \ref{prop1} and the Kato--Rellich theorem. We write
\begin{gather}
\Psi'_\rho\left(|D|^2\right)\langle D\rangle\left(z-H_\rho\right)^{-1}\nonumber\\
=\Psi'_\rho\left(|D|^2\right)\langle D\rangle\left\langle\Psi_\rho\left(|D|^2\right)\right\rangle^{-1}\left\langle\Psi_\rho\left(|D|^2\right)\right\rangle\langle H_\rho\rangle^{-1}\langle H_\rho\rangle\left(z-H_\rho\right)^{-1},\label{the2_5}
\end{gather}
and then estimate
\begin{equation}
\left\|\Psi'_\rho\left(|D|^2\right)\langle D\rangle\left(z-H_\rho\right)^{-1}\right\|\lesssim|{\rm Im}z|^{-1}\langle z\rangle,\label{the2_6}
\end{equation}
if $\rho>1/2$. If $0<\rho\leqslant1/2$, the right-hand side of \eqref{the2_6} can be replaced with just $|{\rm Im}z|^{-1}$ because $\Psi'_\rho(|D|^2)\langle D\rangle$ is bounded. By virtue of the Helffer--Sj\"ostrand formula, \eqref{the2_1}, and \eqref{the2_6}, we have
\begin{gather}
\left[g(H_\rho),\chi\left(\frac{|x|}{2t}\right)\right]=\frac{1}{2\pi{\rm i}}\int_\mathbb{C}\left(\bar{\partial_z}\tilde{g}\right)(z)\left(z-H_\rho\right)^{-1}\nonumber\\
\times\left[\Psi_\rho\left(|D|^2\right),\chi\left(\frac{|x|}{2t}\right)\right]\left(z-H_\rho\right)^{-1}{\rm d}z\wedge{\rm d}\bar{z}=\mathcal{O}\left(t^{-1}\right),\label{the2_7}
\end{gather}
noting that an almost analytic extension $\tilde{g}$ is compactly supported in $\mathbb{C}$. Because we know $\Psi'_\rho(|D|^2)D\cdot(x/|x|)\chi(|x|/(2t))f(H_\rho)=\mathcal{O}(1)$ even though $\rho>1/2$ from \eqref{the2_4} and boundedness of $\Psi'_\rho(|D|^2)Df(H_\rho)$, we have $I_2(t)=\mathcal{O}(t^{-1})$. It follows from \eqref{heisenberg1} and \eqref{the2_3} that
\begin{gather}
\mathbb{D}_{H_\rho}\mathscr{L}(t)=f(H_\rho)\left\{\mathbb{D}_{\Psi_\rho(|D|^2)}X\left(\frac{|x|}{2t}\right)\right\}f(H_\rho)\nonumber\\
\leqslant-\frac{1}{t}\left(\frac{\Theta}{3}-C\right)f(H_\rho)\chi\left(\frac{|x|}{2t}\right)^2f(H_\rho)+\mathcal{O}\left(t^{-2}\right),
\end{gather}
where we put $C=\|g(H_\rho)\Psi'_\rho(|D|^2)D\cdot x/|x|\|$ and choose $\Theta$ such that $\Theta/3-C>0$. This implies \eqref{maximal_velocity_bound} (by \cite[Lemma B.4.1]{DeGe} for example). The proof in the case where $\rho=1$ is given by simply replacing $\Psi'_\rho$ with $1$ in the proof above (see also \cite[Proposition 4.2.1]{DeGe} or \cite[Theorem 2.31]{Iso}). In particular, the commutator calculation is simpler than that of $\rho<1$ because $\Psi_1(|D|^2)=|D|^2$.
\end{proof}

\section{Middle velocity bound}\label{section_middle_velocity_bound}
In Sections \ref{section_middle_velocity_bound} and \ref{section_minimal_velocity_bound}, we assume that the potential function $V$ satisfies Assumption \ref{ass1}. In this section, we focus on proving Theorem \ref{the3} that is the propagation estimate in the mid-range velocity region. This estimate is needed for the proof of Theorem \ref{the1}. The corresponding propagation estimate for the standard two-body Sch\"odinger operator is given in \cite[Proposition 4.4.3]{DeGe} and \cite[Theorem 2.36]{Iso}. To withdraw the time decay in the middle region, we have to add the factor $\Psi'_\rho(|D|^2)D-x/(2t)$ that comes from the Hamilton canonical equation $\nabla_DH_\rho=dx/dt$. Thus, $\Psi'_\rho(|D|^2)D$ and $x/(2t)$ are close asymptotically. In this context, the propagation estimate with regard to the solution of the Hamilton--Jacobi equation was investigated by \cite{Si}.

\begin{The}\label{the3}
\textbf{Middle velocity bound.} For any $0<\theta_1<\theta_2$ and $f\in C_0^\infty(\mathbb{R})$, the inequality
\begin{equation}
\int_1^\infty\left\|F\left(\theta_1\leqslant\frac{|x|}{2t}\leqslant\theta_2\right)\left\{\Psi'_\rho\left(|D|^2\right)D-\frac{x}{2t} \right\}f(H_\rho)e^{-{\rm i}tH_\rho}\phi\right\|_{L^2(\mathbb{R}^n)}^2\frac{{\rm d}t}{t}\lesssim\|\phi\|_{L^2(\mathbb{R}^n)}^2\label{middle_velocity_bound}
\end{equation}
holds for $\phi\in L^2(\mathbb{R}^n)$.
\end{The}

We provide preparations before proving Theorem \ref{the3}. Let $r\in C^\infty(\mathbb{R})$ satisfy that $r(s)=\theta^2/4$ if $s<\theta^2/4$ and $r(s)=s/2$ if $s>\theta^2$ for $0<\theta<\theta_1$, and that $r',r''\geqslant0$ where $r''={\rm d}^2r/{\rm d}s^2$. Putting $R(x)=r(|x|^2)$, we have $R(x)=\theta^2/4$ if $|x|<\theta/2$ and $R(x)=|x|^2/2$ if $|x|>\theta$ holds. We also note that
\begin{equation}
y\cdot\left(\nabla^2R\right)(x)y=4r''\left(|x|^2\right)(x\cdot y)^2+2r'\left(|x|^2\right)|y|^2\geqslant0\label{hessian_matrix}
\end{equation}
holds for any $y\in\mathbb{R}^n$, where $\nabla^2R$ is the Hessian matrix of $R$. The original idea of this function $R$ comes from \cite{DeGe} and \cite{Iso}. We set $\mathscr{M}(t)$ such that
\begin{equation}
\mathscr{M}(t)=\frac{1}{2}\left\{\Psi'_\rho\left(|D|^2\right)D-\frac{x}{2t} \right\}\cdot\left(\nabla R\right)\left(\frac{x}{2t}\right)+{\rm hc}+R\left(\frac{x}{2t}\right).
\end{equation}
We first suppose that $\rho<1$ and the case where $\rho=1$ is given the end of the proof of Theorem \ref{the3}.

\begin{Lem}\label{lem2}
Under the notations above,
\begin{gather}
\mathbb{D}_{\Psi_\rho(|D|^2)}\mathscr{M}(t)\hspace{100mm}\nonumber\\
=\frac{1}{t}\left\{\Psi'_\rho\left(|D|^2\right)D-\frac{x}{2t}\right\}\cdot\left(\nabla^2R\right)\left(\frac{x}{2t}\right)\left\{\Psi'_\rho\left(|D|^2\right)D-\frac{x}{2t}\right\}+\mathcal{O}\left(t^{-2}\right)\qquad\label{heisenberg2}
\end{gather}
holds.
\end{Lem}

\begin{proof}[Proof of Lemm \ref{lem2}]
 In this proof, we use the following commutator notations
\begin{equation}
{\rm ad}_2\left[P,Q\right]=\left[P,\left[P,Q\right]\right],\quad{\rm ad}_3\left[P,Q\right]=\left[P,\left[P,\left[P,Q\right]\right]\right]
\end{equation}
for the operators $P$ and $Q$. By the same computation with \eqref{heisenberg1}, we have
\begin{equation}
\mathbb{D}_{\Psi_\rho(|D|^2)}R\left(\frac{x}{2t}\right)=\frac{1}{2t}\left\{\Psi'_\rho\left(|D|^2\right)D-\frac{x}{2t}\right\}\cdot\left(\nabla R\right)\left(\frac{x}{2t}\right)+{\rm hc}+\mathcal{O}\left(t^{-2}\right).\label{heisenberg3}
\end{equation}
It follows from \eqref{heisenberg3} and
\begin{equation}
\mathbb{D}_{\Psi_\rho(|D|^2)}\frac{x}{2t}=\frac{1}{t}\left\{\Psi'_\rho\left(|D|^2\right)D-\frac{x}{2t}\right\},
\end{equation}
that
\begin{equation}
\mathbb{D}_{\Psi_\rho(|D|^2)}\mathscr{M}(t)=\frac{1}{2}\left\{\Psi'_\rho\left(|D|^2\right)D-\frac{x}{2t}\right\}\cdot\mathbb{D}_{\Psi_\rho(|D|^2)}\left(\nabla R\right)\left(\frac{x}{2t}\right)+{\rm hc}+\mathcal{O}\left(t^{-2}\right).\label{heisenberg4}
\end{equation}
Using \eqref{commutator_left}, we compute, for $1\leqslant j\leqslant n$,
\begin{equation}
{\rm i}\left[\Psi_\rho\left(|D|^2\right),\left(\partial_jR\right)\left(\frac{x}{2t}\right)\right]=\frac{1}{t}\left(\nabla\partial_jR\right)\left(\frac{x}{2t}\right)\cdot\Psi'_\rho\left(|D|^2\right)D+B_j^{\rm L}(t)+\Gamma_j^{\rm L}(t),\label{lem2_1}
\end{equation}
where
\begin{gather}
B_j^{\rm L}(t)=-\frac{\rm i}{4t^2}\left(\Delta\partial_jR\right)\left(\frac{x}{2t}\right)\Psi'_\rho\left(|D|^2\right),\\
\Gamma_j^{\rm L}(t)=\frac{1}{2\pi}\int_\mathbb{C}\bigl(\bar{\partial_z}\tilde{\Psi}_\rho\bigr)(z)\left(z-|D|^2\right)^{-1}\nonumber\\
\times\ {\rm ad}_2\left[|D|^2,\left(\partial_jR\right)\left(\frac{x}{2t}\right)\right]\left(z-|D|^2\right)^{-2}{\rm d}z\wedge{\rm d}\bar{z}.
\end{gather}
Obviously, $B_j^{\rm L}(t)=\mathcal{O}(t^{-2})$ and $\Gamma_j^{\rm L}(t)=\mathcal{O}(t^{-2})$ hold. At the same time, using \eqref{commutator_right}, we obtain another expression
\begin{equation}
{\rm i}\left[\Psi_\rho\left(|D|^2\right),\left(\partial_jR\right)\left(\frac{x}{2t}\right)\right]=\frac{1}{t}\Psi'_\rho\left(|D|^2\right)D\cdot\left(\nabla\partial_jR\right)\left(\frac{x}{2t}\right)-B_j^{\rm R}(t)-\Gamma_j^{\rm R}(t),\label{lem2_2}
\end{equation}
where $B_j^{\rm R}(t)=-B_j^{\rm L}(t)^*$ and
\begin{gather}
\Gamma_j^{\rm R}(t)=\frac{1}{2\pi}\int_\mathbb{C}\bigl(\bar{\partial_z}\tilde{\Psi}_\rho\bigr)(z)\left(z-|D|^2\right)^{-2}\nonumber\\
\times\ {\rm ad}_2\left[|D|^2,\left(\partial_jR\right)\left(\frac{x}{2t}\right)\right]\left(z-|D|^2\right)^{-1}{\rm d}z\wedge{\rm d}\bar{z}.
\end{gather}
with $\Gamma_j^{\rm R}(t)=\mathcal{O}(t^{-2})$. Combining \eqref{heisenberg4}, \eqref{lem2_1}, and \eqref{lem2_2}, we have
\begin{gather}
\mathbb{D}_{\Psi_\rho(|D|^2)}\mathscr{M}(t)=\frac{1}{t}\left\{\Psi'_\rho\left(|D|^2\right)D-\frac{x}{2t}\right\}\cdot\left(\nabla^2R\right)\left(\frac{x}{2t}\right)\left\{\Psi'_\rho\left(|D|^2\right)D-\frac{x}{2t}\right\}\nonumber\\
+\ \frac{1}{2}\left\{\Psi'_\rho\left(|D|^2\right)D-\frac{x}{2t}\right\}\cdot\left\{B^{\rm L}(t)+\Gamma^{\rm L}(t)\right\}\nonumber\\
-\ \frac{1}{2}\left\{B^{\rm R}(t)+\Gamma^{\rm R}(t)\right\}\cdot\left\{\Psi'_\rho\left(|D|^2\right)D-\frac{x}{2t}\right\}+\mathcal{O}\left(t^{-2}\right).\label{heisenberg5}
\end{gather}
We here defined $B^{\rm L}(t)=(B_1^{\rm L}(t),\ldots,B_n^{\rm L}(t))$. $B^{\rm R}$, $\Gamma^{\rm L}$, and $\Gamma^{\rm R}$ have the same definitions. It is clear that
\begin{equation}
\frac{x}{t}\cdot B^{\rm L}(t)=\mathcal{O}\left(t^{-2}\right)\label{lem2_3}
\end{equation}
and that
\begin{gather}
\Psi'_\rho\left(|D|^2\right)D\cdot B^{\rm L}(t)-B^{\rm R}(t)\cdot\Psi'_\rho\left(|D|^2\right)D\nonumber\\
=-\frac{1}{8t^3}\Psi'_\rho\left(|D|^2\right)\left(\Delta^2R\right)\left(\frac{x}{2t}\right)\Psi'_\rho\left(|D|^2\right)=\mathcal{O}\left(t^{-3}\right).\label{lem2_4}
\end{gather}
By calculating the commutator $x_j/t$ and $(z-|D|^2)^{-1}$, we have
\begin{gather}
\left\|\frac{x_j}{t}\left(z-|D|^2\right)^{-1}{\rm ad}_2\left[|D|^2,\left(\partial_jR\right)\left(\frac{x}{2t}\right)\right]\left(z-|D|^2\right)^{-2}\right\|\nonumber\\
\lesssim t^{-2}|{\rm Im}z|^{-3}\langle z\rangle+t^{-3}|{\rm Im}z|^{-4}\langle z\rangle^{3/2}.
\end{gather}
This implies that
\begin{equation}
\frac{x}{t}\cdot \Gamma^{\rm L}(t)=\mathcal{O}\left(t^{-2}\right).\label{lem2_5}
\end{equation}
In the same way, we have
\begin{equation}
\Gamma^{\rm R}(t)\cdot\frac{x}{t}=\mathcal{O}\left(t^{-2}\right).\label{lem2_6}
\end{equation}
If $\rho\leqslant1/2$, clearly
\begin{equation}
\Psi'_\rho\left(|D|^2\right)D\cdot\Gamma^{\rm L}(t)=\mathcal{O}\left(t^{-2}\right)\label{lem2_7}
\end{equation}
because $\Psi'_\rho\left(|D|^2\right)\langle D\rangle$ is bounded. Moreover, even in the case of $1/2<\rho<3/4$, \eqref{lem2_7} holds because
\begin{gather}
\left\|\Psi'_\rho\left(|D|^2\right)\langle D\rangle\left(z-|D|^2\right)^{-1}{\rm ad}_2\left[|D|^2,\left(\partial_jR\right)\left(\frac{x}{2t}\right)\right]\left(z-|D|^2\right)^{-2}\right\|\nonumber\\
\lesssim t^{-2}|{\rm Im}z|^{-3}\langle z\rangle^{\rho+1/2}.
\end{gather}
Similarly, for $\rho<3/4$, we have
\begin{equation}
\Gamma^{\rm R}(t)\cdot\Psi'_\rho\left(|D|^2\right)D=\mathcal{O}\left(t^{-2}\right).\label{lem2_8}
\end{equation}
However, instead of \eqref{lem2_7} and \eqref{lem2_8}, we can have the shaper estimate
\begin{equation}
\Psi'_\rho\left(|D|^2\right)D\cdot\Gamma^{\rm L}(t)-\Gamma^{\rm R}(t)\cdot\Psi'_\rho\left(|D|^2\right)D=\mathcal{O}\left(t^{-3}\right)\label{lem2_9}
\end{equation}
for all $0<\rho<1$ as follows. From the definitions $\Gamma^{\rm L}$ and $\Gamma^{\rm R}$, we denote
\begin{gather}
\Psi'_\rho\left(|D|^2\right)D\cdot \Gamma^{\rm L}(t)-\Gamma^{\rm R}(t)\cdot\Psi'_\rho\left(|D|^2\right)D\nonumber\\
=\frac{1}{2\pi}\int_\mathbb{C}\bigl(\bar{\partial_z}\tilde{\Psi}_\rho\bigr)(z)\left(z-|D|^2\right)^{-1}\sum_{j=1}^nZ_{j,z}(t)\left(z-|D|^2\right)^{-1}{\rm d}z\wedge{\rm d}\bar{z}.
\end{gather}
We here put
\begin{gather}
Z_{j,z}(t)=\Psi'_\rho\left(|D|^2\right)D_j{\rm ad}_2\left[|D|^2,\left(\partial_jR\right)\left(\frac{x}{2t}\right)\right]\left(z-|D|^2\right)^{-1}\nonumber\\
-\left(z-|D|^2\right)^{-1}{\rm ad}_2\left[|D|^2,\left(\partial_jR\right)\left(\frac{x}{2t}\right)\right]\Psi'_\rho\left(|D|^2\right)D_j.
\end{gather}
We further put $Z_{1j,z}$ and $Z_{2j,z}$ by $Z_{j,z}=Z_{1j,z}+Z_{2j,z}$ such that
\begin{align}
Z_{1j,z}(t)&=-\left(z-|D|^2\right)^{-1}{\rm ad}_3\left[|D|^2,\left(\partial_jR\right)\left(\frac{x}{2t}\right)\right]\left(z-|D|^2\right)^{-1}\Psi'_\rho\left(|D|^2\right)D_j,\\
Z_{2j,z}(t)&=\left[\Psi'_\rho\left(|D|^2\right)D_j,{\rm ad}_2\left[|D|^2,\left(\partial_jR\right)\left(\frac{x}{2t}\right)\right]\right]\left(z-|D|^2\right)^{-1}.
\end{align}
We have, from the direct calculation of the commutator,
\begin{equation}
\|\left(z-|D|^2\right)^{-1}Z_{1j,z}(t)\left(z-|D|^2\right)^{-1}\|\lesssim t^{-3}|{\rm Im}z|^{-4}\langle z\rangle^{\rho+1}
\end{equation}
and
\begin{equation}
\int_\mathbb{C}\bigl(\bar{\partial_z}\tilde{\Psi}_\rho\bigr)(z)\left(z-|D|^2\right)^{-1}Z_{1j,z}(t)\left(z-|D|^2\right)^{-1}{\rm d}z\wedge{\rm d}\bar{z}=\mathcal{O}\left(t^{-3}\right).\label{lem2_10}
\end{equation}
As for $Z_{2j,z}$, we write
\begin{equation}
Z_{2j,z}(t)=\left\{\Lambda_{1j}(t)+\Lambda_{2j}(t)\right\}\left(z-|D|^2\right)^{-1},
\end{equation}
using the terms
\begin{align}
\Lambda_{1j}(t)&=\Psi'_\rho\left(|D|^2\right)\left[D_j,{\rm ad}_2\left[|D|^2,\left(\partial_jR\right)\left(\frac{x}{2t}\right)\right]\right],\\
\Lambda_{2j}(t)&=\left[\Psi'_\rho\left(|D|^2\right),{\rm ad}_2\left[|D|^2,\left(\partial_jR\right)\left(\frac{x}{2t}\right)\right]\right]D_j.
\end{align}
We compute directly
\begin{equation}
\|\left(z-|D|^2\right)^{-1}\Lambda_{1j}(t)\left(z-|D|^2\right)^{-2}\|\lesssim t^{-3}|{\rm Im}z|^{-3}\langle z\rangle^\rho.\label{lem2_11}
\end{equation}
$\Lambda_{2j}$ is written such that
\begin{gather}
\Lambda_{2j}(t)=\frac{1}{2\pi{\rm i}}\int_\mathbb{C}\bigl(\bar{\partial}_z\tilde{\Psi}'_\rho\bigr)(z)\left(z-|D|^2\right)^{-1}{\rm ad}_3\left[|D|^2,\left(\partial_jR\right)\left(\frac{x}{2t}\right)\right]\nonumber\\
\times\left(z-|D|^2\right)^{-1}D_j{\rm d}z\wedge{\rm d}\bar{z}\label{lem2_12}
\end{gather}
by the Helffer--Sj\"ostrand formula. The commutator above becomes
\begin{gather}
{\rm ad}_3\left[|D|^2,\left(\partial_jR\right)\left(\frac{x}{2t}\right)\right]=\sum_{k=1}^n{\rm ad}_2\left[D_k,{\rm ad}_2\left[|D|^2,\left(\partial_jR\right)\left(\frac{x}{2t}\right)\right]\right]\nonumber\\
+\ 2\sum_{k=1}^n\left[D_k,{\rm ad}_2\left[|D|^2,\left(\partial_jR\right)\left(\frac{x}{2t}\right)\right]\right]D_k.
\end{gather}
Inserting the estimates
\begin{gather}
\left\|\left(z-|D|^2\right)^{-1}{\rm ad}_2\left[D_k,{\rm ad}_2\left[|D|^2,\left(\partial_jR\right)\left(\frac{x}{2t}\right)\right]\right]\left(z-|D|^2\right)^{-1}\right\|\lesssim t^{-4}|{\rm Im}z|^{-2}\langle z\rangle,\\
\left\|\left(z-|D|^2\right)^{-1}\left[D_k,{\rm ad}_2\left[|D|^2,\left(\partial_jR\right)\left(\frac{x}{2t}\right)\right]\right]\left(z-|D|^2\right)^{-1}\right\|\lesssim t^{-3}|{\rm Im}z|^{-2}\langle z\rangle
\end{gather}
into \eqref{lem2_12}, we have
\begin{equation}
\left\|\left(z-|D|^2\right)^{-1}\Lambda_{2j}(t)\left(z-|D|^2\right)^{-2}\right\|\lesssim t^{-4}|{\rm Im}z|^{-3}\langle z\rangle^{1/2}+t^{-3}|{\rm Im}z|^{-3}\langle z\rangle.\label{lem2_13}
\end{equation}
From \eqref{lem2_11} and \eqref{lem2_13}, we estimate
\begin{equation}
\int_\mathbb{C}\bigl(\bar{\partial_z}\tilde{\Psi}_\rho\bigr)(z)\left(z-|D|^2\right)^{-1}Z_{2j,z}(t)\left(z-|D|^2\right)^{-1}{\rm d}z\wedge{\rm d}\bar{z}=\mathcal{O}\left(t^{-3}\right).\label{lem2_14}
\end{equation}
\eqref{lem2_10} and \eqref{lem2_14} imply \eqref{lem2_9}. In summary, \eqref{heisenberg5}, \eqref{lem2_3}, \eqref{lem2_4}, \eqref{lem2_5}, \eqref{lem2_6}, and \eqref{lem2_9} yield \eqref{heisenberg2}.
\end{proof}

\begin{proof}[Proof of Theorem \ref{the3}]
We take $\chi_1\in C^\infty(\mathbb{R}^n)$ such that $\chi_1(s)=1$ if $s<2\theta_2$ and $\chi_1(s)=0$ if $s>3\theta_2$, and we define the observable $\mathscr{L}(t)$ by
\begin{equation}
\mathscr{L}(t)=f(H_\rho)\chi_1\left(\frac{|x|}{2t}\right)\mathscr{M}(t)\chi_1\left(\frac{|x|}{2t}\right)f(H_\rho),
\end{equation}
according to \cite[ Proposition 4.4.3]{DeGe} and \cite[Theorem 2.36]{Iso}. We know $\mathscr{L}(t)=\mathcal{O}(1)$ because \eqref{the2_4} holds for $\chi_1$. We now compute the Heisenberg derivative of $\mathscr{L}(t)$ associated with $H_\rho$,
\begin{equation}
\mathbb{D}_{H_\rho}\mathscr{L}(t)=I_1(t)+I_2(t)+I_3(t),
\end{equation}
where
\begin{align}
I_1(t)&=f(H_\rho)\left\{\mathbb{D}_{\Psi_\rho(|D|^2)}\chi_1\left(\frac{|x|}{2t}\right)\right\}\mathscr{M}(t)\chi_1\left(\frac{|x|}{2t}\right)f(H_\rho)+{\rm hc},\\
I_2(t)&=f(H_\rho)\chi_1\left(\frac{|x|}{2t}\right)\left\{\mathbb{D}_{\Psi_\rho(|D|^2)}\mathscr{M}(t)\right\}\chi_1\left(\frac{|x|}{2t}\right)f(H_\rho),\\
I_3(t)&=f(H_\rho)\chi_1\left(\frac{|x|}{2t}\right){\rm i}\left[V,\mathscr{M}(t)\right]\chi_1\left(\frac{|x|}{2t}\right)f(H_\rho).
\end{align}

\noindent
{\bf Estimate for \boldmath{$I_1$}.}\quad The same computation with \eqref{lem2_1} and \eqref{lem2_2} give
\begin{gather}
I_1(t)=\frac{1}{t}f(H_\rho)\left\{\Psi'_\rho\left(|D|^2\right)D-\frac{x}{2t}\right\}\cdot\frac{x}{|x|}\chi'_1\left(\frac{|x|}{2t}\right)\mathscr{M}(t)\chi_1\left(\frac{|x|}{2t}\right)f(H_\rho)\nonumber\\
+\ {\rm hc}+\mathcal{O}\left(t^{-2}\right).
\end{gather}
Let $\chi_2\in C_0^\infty(\mathbb{R})$ such that $\chi_2(s)=1$ if $2\theta_2<s<3\theta_2$ and $\chi_2(s)=0$ if $s<\theta_2$ and $s>4\theta_2$. We see that $\chi_2$ satisfies $\chi'_1=\chi_2^2\chi'_1$. Let $g\in C_0^\infty(\mathbb{R})$ such that $f=fg$. We compute
\begin{gather}
f(H_\rho)\left\{\Psi'_\rho\left(|D|^2\right)D-\frac{x}{2t}\right\}\cdot\frac{x}{|x|}\chi'_1\left(\frac{|x|}{2t}\right)\mathscr{M}(t)\chi_1\left(\frac{|x|}{2t}\right)f(H_\rho)\nonumber\\
=f(H_\rho)\chi_2\left(\frac{|x|}{2t}\right)g(H_\rho)\left\{\Psi'_\rho\left(|D|^2\right)D-\frac{x}{2t}\right\}\cdot\frac{x}{|x|}\chi'_1\left(\frac{|x|}{2t}\right)\nonumber\\
\times\ \mathscr{M}(t)\chi_1\left(\frac{|x|}{2t}\right)\chi_2\left(\frac{|x|}{2t}\right)f(H_\rho)+\mathcal{O}\left(t^{-1}\right).\label{the3_1}
\end{gather}
We here used the commutator estimates \eqref{the2_4} and \eqref{the2_7}. Because $4\rho-2<2\rho$ and $\Psi'_\rho(|D|^2)^2|D|^2g(H_\rho)=\langle D\rangle^{4\rho-4}|D|^2g(H_\rho)$ is bounded, we have
\begin{equation}
I_1(t)=\frac{1}{t}f(H_\rho)\chi_2\left(\frac{|x|}{2t}\right)\mathcal{O}(1)\chi_2\left(\frac{|x|}{2t}\right)f(H_\rho)+\mathcal{O}\left(t^{-2}\right).\label{the3_2}
\end{equation}
If necessary, we can assume that $\theta_2$ is sufficiently large. By virtue of \eqref{the3_2} and Theorem \ref{the2},
\begin{equation}
\int_1^\infty\left|\left(I_1(t)e^{-{\rm i}tH_\rho}\phi,e^{-{\rm i}tH_\rho}\phi\right)_{L^2}\right|dt\lesssim\int_1^\infty\left\|\chi_2\left(\frac{|x|}{2t}\right)f(H_\rho)e^{-{\rm i}tH_\rho}\phi\right\|_{L^2}^2\frac{{\rm d}t}{t}\lesssim\|\phi\|_{L^2}^2\label{the3_3}
\end{equation}
is obtained, where $(\cdot,\cdot)_{L^2}$ is the scalar product of $L^2(\mathbb{R}^n)$.\\
\noindent
{\bf Estimate for \boldmath{$I_2$}.}\quad We take $\chi\in C_0^\infty(\mathbb{R})$ such that $\chi(s)=1$ if $\theta_1\leqslant s\leqslant\theta_2$ and $\chi(s)=0$ if $s<(\theta_1+\theta)/2$ and $s>\theta_2+(\theta_1-\theta)/2$. Noting that $(\nabla^2R)(x)={\rm Id}$, which is the identity matrix if $|x|\geqslant(\theta_1+\theta)/2$, and that $\nabla^2R$ is non-negative from \eqref{hessian_matrix}, we have
\begin{gather}
\left(\nabla^2R\right)\left(\frac{x}{2t}\right)=\chi\left(\frac{|x|}{2t}\right)\left(\nabla^2R\right)\left(\frac{x}{2t}\right)\chi\left(\frac{|x|}{2t}\right)\hspace{30mm}\nonumber\\
+\ \sqrt{1-\chi\left(\frac{|x|}{2t}\right)^2}\left(\nabla^2R\right)\left(\frac{x}{2t}\right)\sqrt{1-\chi\left(\frac{|x|}{2t}\right)^2}\geqslant\chi\left(\frac{|x|}{2t}\right)^2{\rm Id}.\label{the3_4}
\end{gather}
Using \eqref{the2_4}, \eqref{heisenberg2}, \eqref{the3_4}, and $\chi_1\chi=\chi$, $I_2$ is estimated as
\begin{gather}
I_2(t)\geqslant\frac{1}{t}f(H_\rho)\left\{\Psi'_\rho\left(|D|^2\right)D-\frac{x}{2t}\right\}\nonumber\\
\cdot\chi\left(\frac{|x|}{2t}\right)^2\left\{\Psi'_\rho\left(|D|^2\right)D-\frac{x}{2t}\right\}f(H_\rho)+\mathcal{O}\left(t^{-2}\right).\label{the3_5}
\end{gather}

\noindent
{\bf Estimate for \boldmath{$I_3$}.}\quad It follows that $(\nabla V_{\rm long})(x)\cdot(\nabla R)(x/(2t))=\mathcal{O}(t^{-1-\gamma_{\rm long}})$ by the condition \eqref{long_range_decay} because $|x|\geqslant t\theta$ holds on the support of $(\partial_jR)(x/(2t))$ for all $1\leqslant j\leqslant n$. We thus compute
\begin{gather}
\left[V_{\rm long},\Psi'_\rho\left(|D|^2\right)D\cdot\left(\nabla R\right)\left(\frac{x}{2t}\right)\right]\nonumber\\
=\Psi'_\rho\left(|D|^2\right)\mathcal{O}\left(t^{-1-\gamma_{\rm long}}\right)+\left[V_{\rm long},\Psi'_\rho\left(|D|^2\right)\right]\left\{\left(\nabla R\right)\left(\frac{x}{2t}\right)\cdot D-\frac{\rm i}{2t}\left(\Delta R\right)\left(\frac{x}{2t}\right)\right\}.\label{the3_6}
\end{gather}
To apply the Helffer--Sj\"ostrand formula, we compute
\begin{gather}
\left(z-|D|^2\right)^{-1}\left(\nabla R\right)\left(\frac{x}{2t}\right)\cdot D=\left(\nabla R\right)\left(\frac{x}{2t}\right)\cdot D\left(z-|D|^2\right)^{-1}\nonumber\\
+\left(z-|D|^2\right)^{-1}\left[|D|^2,\left(\nabla R\right)\left(\frac{x}{2t}\right)\cdot D\right]\left(z-|D|^2\right)^{-1}.
\end{gather}
Noting that $[|D|^2,V_{\rm long}]=-{\rm i}D\cdot\nabla V_{\rm long}-{\rm i}\nabla V_{\rm long}\cdot D$, we have the estimate
\begin{gather}
\left\|\left(z-|D|^2\right)^{-1}\left[|D|^2,V_{\rm long}\right]\left(\nabla R\right)\left(\frac{x}{2t}\right)\cdot D\left(z-|D|^2\right)^{-1}\right\|\nonumber\\
\lesssim t^{-1-\gamma_{\rm long}}|{\rm Im}z|^{-2}\langle z\rangle+t^{-2-\gamma_{\rm long}}|{\rm Im}z|^{-2}\langle z\rangle^{1/2},\label{the3_7}
\end{gather}
and, by
\begin{equation}
\left[|D|^2,\left(\nabla R\right)\left(\frac{x}{2t}\right)\cdot D\right]=-\frac{\rm i}{t}\left(\nabla^2 R\right)\left(\frac{x}{2t}\right)D\cdot D-\frac{\rm 1}{4t^2}\left(\nabla\Delta R\right)\left(\frac{x}{2t}\right)\cdot D,\label{the3_8}
\end{equation}
we also have
\begin{gather}
\left\|\left(z-|D|^2\right)^{-1}\left[|D|^2,V_{\rm long}\right]\left(z-|D|^2\right)^{-1}\left[|D|^2,\left(\nabla R\right)\left(\frac{x}{2t}\right)\cdot D\right]\left(z-|D|^2\right)^{-1}\right\|\nonumber\\
\lesssim t^{-2-\gamma_{\rm long}}|{\rm Im}z|^{-3}\langle z\rangle^{3/2}+t^{-3-\gamma_{\rm long}}|{\rm Im}z|^{-3}\langle z\rangle+t^{-2}|{\rm Im}z|^{-4}\langle z\rangle^2+t^{-2}|{\rm Im}z|^{-3}\langle z\rangle.\label{the3_9}
\end{gather}
We here computed the commutator $(z-|D|^2)^{-1}$ and $(\nabla^2R)(x/(2t))D\cdot D$. \eqref{the3_7} and \eqref{the3_9} imply
\begin{equation}
\left[V_{\rm long},\Psi'_\rho\left(|D|^2\right)\right]\left(\nabla R\right)\left(\frac{x}{2t}\right)\cdot D=\mathcal{O}\left(t^{-1-\gamma_{\rm long}}\right)+\mathcal{O}\left(t^{-2}\right)\label{the3_10}
\end{equation}
by the Helffer--Sj\"ostrand formula. By the same computations, we have
\begin{equation}
\left[V_{\rm long},\Psi'_\rho\left(|D|^2\right)\right]\left(\Delta R\right)\left(\frac{x}{2t}\right)=\mathcal{O}\left(t^{-1-\gamma_{\rm long}}\right)+\mathcal{O}\left(t^{-2}\right).\label{the3_11}
\end{equation}
\eqref{the3_6}, \eqref{the3_10}, and \eqref{the3_11} yield
\begin{equation}
\left[V_{\rm long},\mathscr{M}(t)\right]=\mathcal{O}\left(t^{-1-\gamma_{\rm long}}\right)+\mathcal{O}\left(t^{-2}\right).\label{the3_12}
\end{equation}
We put
\begin{equation}
\mathscr{K}(t)=\frac{1}{2}\left\{\Psi'_\rho\left(|D|^2\right)D-\frac{x}{2t} \right\}\cdot\left(\nabla R\right)\left(\frac{x}{2t}\right)+{\rm hc}.\label{the3_13}
\end{equation}
Because we know that $\langle x\rangle^{\gamma_{\rm sing}}V_{\rm sing}\chi_1(x/(2t))f(H_\rho)=\mathcal{O}(1)$ by Proposition \ref{prop1} and \eqref{the2_7} or \eqref{lem2_2}, we write
\begin{gather}
f(H_\rho)\chi_1\left(\frac{|x|}{2t}\right)\left[V_{\rm sing},\mathscr{M}(t)\right]\chi_1\left(\frac{|x|}{2t}\right)f(H_\rho)\nonumber\\
=\mathcal{O}\left(1\right)\langle x\rangle^{-\gamma_{\rm sing}}\mathscr{K}(t)\chi_1\left(\frac{|x|}{2t}\right)f(H_\rho)-f(H_\rho)\chi_1\left(\frac{|x|}{2t}\right)\mathscr{K}(t)\langle x\rangle^{-\gamma_{\rm sing}}\mathcal{O}\left(1\right).\label{the3_14}
\end{gather}
By computing the commutator $\Psi'_\rho(|D|^2)D_j$ and $(\partial_jR)(x/(2t))$, we have 
\begin{gather}
\langle x\rangle^{-\gamma_{\rm sing}}\Psi'_\rho\left(|D|^2\right)D_j\left(\partial_jR\right)\left(\frac{x}{2t}\right)=\mathcal{O}(t^{-\gamma_{\rm sing}})\Psi'_\rho\left(|D|^2\right)D_j\nonumber\\
+\ \mathcal{O}(t^{-1-\gamma_{\rm sing}})+\langle x\rangle^{-\gamma_{\rm sing}}D_j\left[\Psi'_\rho\left(|D|^2\right),\left(\partial_jR\right)\left(\frac{x}{2t}\right)\right].\label{the3_15}
\end{gather}
It follows from
\begin{gather}
\left\|\langle x\rangle^{-\gamma_{\rm sing}}\left(z-|D|^2\right)^{-1}D_j\left[|D|^2,\left(\partial_jR\right)\left(\frac{x}{2t}\right)\right]\left(z-|D|^2\right)^{-1}\right\|\qquad\nonumber\\
\lesssim t^{-1-\gamma_{\rm sing}}|{\rm Im}z|^{-2}\langle z\rangle+t^{-2}|{\rm Im}z|^{-3}\langle z\rangle^{3/2}+t^{-2}|{\rm Im}z|^{-2}\langle z\rangle^{1/2}
\end{gather}
as in \eqref{the3_9} that
\begin{equation}
\langle x\rangle^{-\gamma_{\rm sing}}D_j\left[\Psi'_\rho\left(|D|^2\right),\left(\partial_jR\right)\left(\frac{x}{2t}\right)\right]=\mathcal{O}\left(t^{-2}\right)\label{the3_16}
\end{equation}
by the Helffer--Sj\"ostrand formula again. \eqref{the3_15} and \eqref{the3_16} imply that
\begin{equation}
\langle x\rangle^{-\gamma_{\rm sing}}\mathscr{K}(t)\chi_1\left(\frac{|x|}{2t}\right)f(H_\rho)=\mathcal{O}\left(t^{-\gamma_{\rm sing}}\right)+\mathcal{O}\left(t^{-2}\right)
\end{equation}
and that, from \eqref{the3_14},
\begin{equation}
f(H_\rho)\chi_1\left(\frac{|x|}{2t}\right)\left[V_{\rm sing},\mathscr{M}(t)\right]\chi_1\left(\frac{|x|}{2t}\right)f(H_\rho)=\mathcal{O}\left(t^{-\gamma_{\rm sing}}\right)+\mathcal{O}\left(t^{-2}\right).\label{the3_17}
\end{equation}
We also have
\begin{equation}
f(H_\rho)\chi_1\left(\frac{|x|}{2t}\right)\left[V_{\rm short},\mathscr{M}(t)\right]\chi_1\left(\frac{|x|}{2t}\right)f(H_\rho)=\mathcal{O}\left(t^{-\gamma_{\rm short}}\right)+\mathcal{O}\left(t^{-2}\right)\label{the3_18}
\end{equation}
by replacing $\langle x\rangle^{-\gamma_{\rm sing}}$ with $V_{\rm short}$ in the computations above. By \eqref{the3_12}, \eqref{the3_17}, and \eqref{the3_18}, we have
\begin{equation}
\label{the3_19}
I_3(t)=\mathcal{O}\left(t^{-\min\{\gamma_{\rm sing},\gamma_{\rm short},1+\gamma_{\rm long},2\}}\right).
\end{equation}
\par
We combine \eqref{the3_5} and \eqref{the3_19}. There exists a constant $C>0$ such that
\begin{gather}
\left(\left\{\mathbb{D}_{H_\rho}\mathscr{L}(t)\right\}e^{-{\rm i}tH_\rho}\phi,e^{-{\rm i}tH_\rho}\phi\right)_{L^2}\nonumber\\
\geqslant\frac{1}{t}\left\|\chi\left(\frac{|x|}{2t}\right)\left\{\Psi'_\rho\left(|D|^2\right)D-\frac{x}{2t}\right\}f(H_\rho)e^{-{\rm i}tH_\rho}\phi\right\|_{L^2}^2\nonumber\\
-\left|\left(I_1(t)e^{-{\rm i}tH_\rho}\phi,e^{-{\rm i}tH_\rho}\phi\right)_{L^2}\right|-Ct^{-\min\{\gamma_{\rm sing},\gamma_{\rm short},1+\gamma_{\rm long},2\}}\|\phi\|_{L^2}^2
\end{gather}
holds. This completes our proof for the case $\rho<1$ by virtue of \eqref{the3_3} and $\min\{\gamma_{\rm sing},\gamma_{\rm short},1+\gamma_{\rm long},2\}>1$. In the case where $\rho=1$, the proof is simpler (see \cite[Proposition 4.4.3]{DeGe} or \cite[Theorem 2.36]{Iso}). Indeed, by replacing $\Psi'_\rho$ with $1$, we omit many of the commutator calculations. In particular, \eqref{heisenberg3} holds without the error term $\mathcal{O}(t^{-2})$. We therefore explicitly have \eqref{heisenberg2} replacing $\mathcal{O}(t^{-2})$ with $-(\Delta^2R)(x/(2t))/(16t^3)$.
\end{proof}

\section{Minimal velocity bound}\label{section_minimal_velocity_bound}
This section completes the proof of Theorem \ref{the1}. Before giving the proof, we initially prepare the Mourre estimate of our version in Theorem \ref{the4} and prove the isolatedness and finite multiplicity of $\sigma_{\rm pp}(H_\rho)\setminus\{0\}$ in Corollary \ref{cor1}. When we consider the Mourre estimate, how to choose a conjugate operator is the heart of matter. In our case, we first employ
\begin{gather}
\hat{A}_\rho=\frac{\rho}{2}\left\{\langle D\rangle^{2\rho-2}D\cdot x+x\cdot D\langle D\rangle^{2\rho-2}\right\}\nonumber\\
=\frac{1}{2}\left\{\Psi'_\rho\left(|D|^2\right)D\cdot x+x\cdot D\Psi'_\rho\left(|D|^2\right)\right\}\label{conjugate1}
\end{gather}
motivated with which
\begin{equation}
{\rm i}\left[\Psi_\rho\left(|D|^2\right),\hat{A}_\rho\right]=2\Psi'_\rho\left(|D|^2\right)^2|D|^2\label{positive_commutator}
\end{equation}
holds by a straightforward computation on $C_0^\infty(\mathbb{R}^n)$ and \eqref{positive_commutator} is non-negative. The choice of conjugate operator is not unique. Indeed, if $1/2\leqslant\rho\leqslant1$, we can admit
\begin{equation}
A=\hat{A}_1=\frac{1}{2}\left(D\cdot x+x\cdot D \right)\label{conjugate_dilation}
\end{equation}
(see Remarks \ref{rem5} and \ref{rem6}) that works well for the standard Schr\"odinger operator.

The resolvent of $|D|^2$ was first introduced as the conjugate operator of the Mourre estimate by \cite{Yo} such that
\begin{equation}
\frac{1}{2}\left\{\langle D\rangle^{-2}D\cdot x+x\cdot D\langle D\rangle^{-2}\right\}
\end{equation}
to consider the time-dependent Schr\"odinger operator
\begin{equation}
H(t)=-\Delta+V(t)\label{time_periodic}
\end{equation}
where $V(t)=V(t,x)$ had time-periodicity in $t$. Thereafter, \cite{AdKi} also treated the Hamiltonian \eqref{time_periodic} and introduced the resolvent of $D_t=-{\rm i}{\rm d}/{\rm d}t$ into the conjugate operator to relax the smoothness condition on $V$. Both \cite{AdKi} and \cite{Yo} applied the Howland--Yajima method for the Floquet Hamiltonian $D_t+H(t)$. They estimated the commutators with the Floquet Hamiltonian and the conjugate operator, and their estimates were independent of the fractional operator.

We now begin with the self-adjointness of $\hat{A}_\rho$.

\begin{Prop}\label{prop2}
$\hat{A}_\rho$ is essentially self-adjoint with the core $C_0^\infty(\mathbb{R}^n)$.
\end{Prop}

\begin{proof}
We define the operator $N_\rho$ by
\begin{equation}
N_\rho=\Psi'_\rho\left(|D|^2\right)^2|D|^2+|x|^2+1.\label{nelson}
\end{equation}
If $\rho>1/2$, $N_{\rho}$ is self-adjoint on $H^{4\rho-2}(\mathbb{R}^n)\cap\dom|x|^2$. Whereas if $\rho\leqslant1/2$, $\Psi'_\rho(|D|^2)^2|D|^2$ is bounded and $N_{\rho}$ is self-adjoint on $\dom|x|^2$. We compute on $C_0^\infty(\mathbb{R}^n)$,
\begin{gather}
{\rm i}\left[\hat{A}_\rho,\Psi'_\rho\left(|D|^2\right)^2|D|^2\right]=-2\left\{2\Psi''_\rho\left(|D|^2\right)|D|^2+\Psi'_\rho\left(|D|^2\right)\right\}\Psi'_\rho\left(|D|^2\right)^2|D|^2\nonumber\\
\lesssim\Psi'_\rho\left(|D|^2\right)^2|D|^2\label{prop2_1}
\end{gather}
because $\Psi''_\rho(|D|^2)|D|^2$ and $\Psi'_\rho(|D|^2)$ are bounded. In the rest of this proof, we put $D_{\rho j}=\Psi'_\rho(|D|^2)D_j$ for simplicity. We thus compute, for $1\leqslant j,k\leqslant n$, 
\begin{equation}
{\rm i}\left[D_{\rho j}x_j+x_jD_{\rho j},x_k^2\right]=2x_j{\rm i}\left[D_{\rho j},x_k\right]x_k+{\rm hc}+{\rm i}\left[\left[\left[D_{\rho j},x_k\right],x_k\right],x_j\right].\label{prop2_2}
\end{equation}
Because $[D_{\rho j},x_k]$ and $[[[D_{\rho j},x_k],x_k],x_j]$ are bounded, \eqref{prop2_2} implies
\begin{equation}
{\rm i}\bigl[\hat{A}_\rho,|x|^2\bigr]\lesssim|x|^2+1,\label{prop2_3}
\end{equation}
where we used the estimate
\begin{gather}
\left|\left(x_j{\rm i}\left[D_{\rho j},x_k\right]x_k\phi,\phi\right)\right|\leqslant\left\|{\rm i}\left[D_{\rho j},x_k\right]\right\|\|x_j\phi\|\|x_k\phi|\|
\lesssim|\|x_j\phi\|^2+|\|x_k\phi\|^2
\end{gather}
for $\phi\in C_0^\infty(\mathbb{R}^n)$. It follows from \eqref{prop2_1} and \eqref{prop2_3} that
\begin{equation}
{\rm i}\bigl[\hat{A}_\rho,N_\rho\bigr]\lesssim N_\rho.\label{prop2_4}
\end{equation}
We next compute, noting that $[D_{\rho j},x_j]$ dose not depend on $x_j$,
\begin{gather}
\left(D_{\rho j}x_j+x_jD_{\rho j}\right)^2=2D_{\rho j}^2x_j^2+2x_j^2D_{\rho j}^2-2D_{\rho j}\left[\left[D_{\rho j},x_j\right],x_j\right]-3\left[D_{\rho j},x_j\right]^2
\end{gather}
and we have
\begin{gather}
2\hat{A}_\rho^2=\sum_{j=1}^n\left(D_{\rho j}^2x_j^2+x_j^2D_{\rho j}^2\right)+\frac{1}{2}\sum_{j=1, k\not=j}^n\left(D_{\rho j}x_j+x_jD_{\rho j}\right)\left(D_{\rho k}x_k+x_kD_{\rho k}\right)\nonumber\\
-\ \frac{1}{2}\sum_{j=1}^n\left\{2D_{\rho j}\left[\left[D_{\rho j},x_j\right],x_j\right]+3\left[D_{\rho j},x_k\right]^2\right\}.\label{prop2_5}
\end{gather}
We here note that $D_{\rho j}[[D_{\rho j},x_j],x_j]$ is bounded. We also compute
\begin{gather}
N_\rho^2\geqslant\sum_{j=1}^n\left(D_{\rho j}^2x_j^2+x_j^2D_{\rho j}^2\right)+\sum_{j=1, k\not=j}^n\left(D_{\rho j}^2+x_j^2\right)\left(D_{\rho k}^2+x_k^2\right)+1.\label{prop2_6}
\end{gather}
We have
\begin{gather}
2\left(D_{\rho j}^2+x_j^2\right)\left(D_{\rho k}^2+x_k^2\right)-\left(D_{\rho j}x_j+x_jD_{\rho j}\right)\left(D_{\rho k}x_k+x_kD_{\rho k}\right)\nonumber\\
=\left(D_{\rho j}D_{\rho k}-x_jx_k\right)^2+\left(D_{\rho j}x_k-x_jD_{\rho k}\right)^2\nonumber\\
+D_{\rho j}^2D_{\rho k}^2+x_j^2x_k^2+D_{\rho j}^2x_k^2+x_j^2D_{\rho k}^2+R_{jk}\geqslant R_{jk},\label{prop2_7}
\end{gather}
where
\begin{gather}
\sum_{j=1, k\not=j}^nR_{jk}=4{\rm i}\sum_{j=1, k\not=j}^n\left\{x_j\Psi''_\rho\left(|D|^2\right)D_jD_kD_{\rho k}-D_{\rho j}\Psi''_\rho\left(|D|^2\right)D_jD_kx_k\right\}\nonumber\\
=4{\rm i}\sum_{j=1, k\not=j}^n\left[x_j, \Psi'_\rho\left(|D|^2\right)\Psi''_\rho\left(|D|^2\right)D_jD_k^2\right]\label{prop2_8}
\end{gather}
is bounded. From \eqref{prop2_5}, \eqref{prop2_6}, \eqref{prop2_7}, and \eqref{prop2_8}, it follows that
\begin{equation}
\hat{A}_\rho^2\lesssim N_\rho^2.\label{prop2_9}
\end{equation}
By \eqref{prop2_4} and \eqref{prop2_9}, the Nelson commutator theorem \cite[Theorem X.37]{ReSi} completes our proof.
\end{proof}

\begin{Lem}\label{lem3}
For $z\in\mathbb{C}\setminus\mathbb{R}$, the relation
\begin{equation}
\dom\langle x\rangle\subset\left\{\phi\in\dom\hat{A}_\rho\bigm|\left(z-H_\rho\right)^{-1}\phi\in\dom\hat{A}_\rho \right\}\label{lem3_1}
\end{equation}
holds.
\end{Lem}

\begin{proof}
We prove the domain property
\begin{equation}
\left(z-H_{\rho}\right)^{-1}\dom\langle x\rangle\subset\dom\hat{A}_\rho\label{lem3_2}
\end{equation}
that is equivalent to \eqref{lem3_1}. We first prove that
\begin{equation}
\langle x\rangle\left(z-H_\rho\right)^{-1}\langle x\rangle^{-1}\label{lem3_3}
\end{equation}
is bounded. By the resolvent formula, we write
\begin{equation}
\left(z-H_\rho\right)^{-1}=\left\{z-\Psi_\rho\left(|D|^2\right)\right\}^{-1}V\left(z-H_\rho\right)^{-1}+\left\{z-\Psi_\rho\left(|D|^2\right)\right\}^{-1}
\end{equation}
It follows from
\begin{equation}
\left[x_j,\left\{z-\Psi_\rho\left(|D|^2\right)\right\}^{-1}\right]=2{\rm i}\Psi'_\rho\left(|D|^2\right)D_j\left\{z-\Psi_\rho\left(|D|^2\right)\right\}^{-2}\label{lem3_4}
\end{equation}
on $\dom\langle x\rangle$ for $1\leqslant j\leqslant n$ that
\begin{equation}
\langle x\rangle^{\nu}\left\{z-\Psi_\rho\left(|D|^2\right)\right\}^{-1}\langle x\rangle^{-\nu}\label{lem3_5}
\end{equation}
is bounded for $\nu\in\mathbb{R}$ by \eqref{lem3_4} and the complex interpolation derived from the Hadamard three-line theorem (\cite{ReSi}, Appendix to IX.4). Because $\langle x\rangle(V_{\rm sing}+V_{\rm short})(z-H_\rho)^{-1}$ and \eqref{lem3_5} of $\nu=1$ are bounded, to prove the boundedness of \eqref{lem3_3}, it suffices to prove that $\langle x\rangle V_{\rm long}(z-H_\rho)^{-1}\langle x\rangle^{-1}$ is bounded. Using the resolvent formula, we have
\begin{gather}
V_{\rm long}\left(z-H_\rho\right)^{-1}\nonumber\\
=V_{\rm long}\left\{z-\Psi_\rho\left(|D|^2\right)\right\}^{-1}V_{\rm long}\left(z-H_\rho\right)^{-1}+V_{\rm long}\left\{z-\Psi_\rho\left(|D|^2\right)\right\}^{-1}.\label{lem3_6}
\end{gather}
If $\gamma_{\rm long}\geqslant1/2$, writing $\langle x\rangle^{2\gamma_{\rm long}}V_{\rm long}\{z-\Psi_\rho(|D|^2)\}^{-1}V_{\rm long}$ such that
\begin{gather}
\langle x\rangle^{2\gamma_{\rm long}}V_{\rm long}\left\{z-\Psi_\rho(|D|^2)\right\}^{-1}V_{\rm long}\nonumber\\
=\langle x\rangle^{\gamma_{\rm long}}V_{\rm long}\langle x\rangle^{\gamma_{\rm long}}\left\{z-\Psi_\rho\left(|D|^2\right)\right\}^{-1}\langle x\rangle^{-\gamma_{\rm long}}\langle x\rangle^{\gamma_{\rm long}}V_{\rm long},\label{lem3_7}
\end{gather}
we find that \eqref{lem3_7} is bounded by \eqref{lem3_5} of $\nu=\gamma_{\rm long}$ and that $\langle x\rangle V_{\rm long}(z-H_\rho)^{-1}\langle x\rangle^{-1}$ is bounded by \eqref{lem3_6}. For the general $\gamma_{\rm long}>0$, we can take $N\in\mathbb{N}$ that satisfies $\gamma_{\rm long}\geqslant 1/N>0$ and iterate the above procedure $N-1$ times. If $\rho\leqslant1/2$, the boundedness of \eqref{lem3_3} implies \eqref{lem3_2} immediately because $\Psi'_\rho(|D|^2)\langle D\rangle$ is bounded and $\hat{A}_\rho$ is closed. If $\rho>1/2$, we can also prove that $\langle x\rangle\langle D\rangle^{2\rho-1}(z-H_\rho)^{-1}\langle x\rangle^{-1}$ is bounded in the same way, noting that
\begin{equation}
\langle x\rangle\langle D\rangle^{2\rho-1}\left\{z-\Psi_\rho\left(|D|^2\right)\right\}^{-1}\langle x\rangle^{-1}
\end{equation}
is bounded. We thus have \eqref{lem3_2} even for $1/2\leqslant\rho\leqslant1$. In more detail, because $\langle x\rangle(z-H_\rho)^{-1}\phi\in H^{2\rho-1}(\mathbb{R}^n)$ for $\phi\in\dom\langle x\rangle$, there exists a sequence $\psi_k\in C_0^\infty(\mathbb{R}^n)$ such that $\langle x\rangle\psi_k\rightarrow\langle x\rangle(z-H_\rho)^{-1}\phi$ as $k\rightarrow\infty$ in $H^{2\rho-1}(\mathbb{R}^n)$. We have $\hat{A}_\rho\psi_k\rightarrow\hat{A}_\rho(z-H_\rho)^{-1}\phi$ as $k\rightarrow\infty$ and $(z-H_\rho)^{-1}\phi\in\dom\hat{A_\rho}$ noting that $\hat{A}_\rho$ is closed.
\end{proof}

By Proposition \ref{prop1}, $V$ is relatively compact associated with $\Psi_\rho(|D|^2)$. This can be proved in the same way as in the standard Schr\"odinger case. Because the essential spectrum of $\Psi_\rho(|D|^2)$ is $[0,\infty)$, the essential spectrum of $H_\rho$ is also coincident with $[0,\infty)$ by virtue of the relative compactness of $V$ and the Weyl theorem (\cite[Theorem XIII.14]{ReSi}). 

We now prove the Mourre estimate. However, it seems difficult that the commutator \eqref{positive_commutator} extends on $H^{2\rho}(\mathbb{R}^n)\cap\dom\hat{A}_\rho$ in the form sense. To overcome this difficulty, we give a modification in $\hat{A}_\rho$ according to the original idea \cite{Mo}. Let $G_\rho\in C_0^\infty(\mathbb{R})$ such that $G_\rho(s)=\Psi'_\rho(s)$ on a some compact set of $\mathbb{R}$. We define
\begin{equation}
A_\rho=\frac{1}{2}\left\{G_\rho\left(|D|^2\right)D\cdot x+x\cdot G_\rho\left(|D|^2\right)D\right\}.\label{conjugate2}
\end{equation}
By the same way with Proposition \ref{prop2}, $A_\rho$ is essentially self-adjoint with the core $C_0^\infty(\mathbb{R}^n)$ and Lemma \ref{lem3} also holds even for $A_\rho$. In particular, it follows from the proof of Lemma \ref{lem3} that
\begin{equation}
2A_\rho\phi=G_\rho\left(|D|^2\right)D\cdot(x\phi)+x\cdot\left\{G_\rho\left(|D|^2\right)D\phi\right\}
\end{equation}
for $\phi\in\dom\langle x\rangle$. This will be often used in the rest of our discussion.

\begin{The}\label{the4}
\textbf{Mourre estimate.}
Let $0<\lambda_1<\lambda_2$ and $g\in C_0^\infty((\lambda_1,\lambda_2))$. Assume that $G_\rho(s)=\Psi'_\rho(s)$ if $\Psi_\rho(s)\in\supp g$. There exists a compact operator $K$ such that
\begin{equation}
g(H_\rho){\rm i}\left[H_\rho,A_\rho\right]_{-2\rho}g(H_\rho)\geqslant\frac{2\rho^2\lambda_1}{(1+\lambda_2)^{(1-\rho)/\rho}}g(H_\rho)^2+K\label{mourre_inequality1}
\end{equation}
holds, where the sense of the extended commutator $[H_\rho,A_\rho]_{-2\rho}$ is explained in the proof.
\end{The}

\begin{proof}
We first suppose that $\rho<1$. By \cite[Proposition II.1]{Mo}, the form commutator ${\rm i}[\Psi_\rho(|D|^2),A_\rho]$ on $H^{2\rho}(\mathbb{R}^n)\cap\dom|x| ^2$ is extended on $H^{2\rho}(\mathbb{R}^n)\cap\dom A_\rho$ and there exists the self-adjoint operator ${\rm i}[\Psi_\rho(|D|^2),A_\rho]^0$ associated with the closed extension of ${\rm i}[\Psi_\rho(|D|^2),A_\rho]$ (see also the proof of \cite[Corollary I.3]{Mo}). We therefore have
\begin{equation}
{\rm i}\left[\Psi_\rho\left(|D|^2\right),A_\rho\right]^0=2G_\rho(|D|^2)^2|D|^2
\end{equation}
that is a bounded operator. Using the fact that $\langle x\rangle^{-1}A_\rho$ is bounded and \eqref{relatively_bound1}, we estimate
\begin{gather}
\left|\left(A_\rho\phi,V_{\rm sing}\psi\right)_{L^2}\right|=\left|\left(\langle x\rangle^{-1}A_\rho\phi,\langle x\rangle V_{\rm sing}\psi\right)_{L^2}\right|\nonumber\\
\lesssim\left\|\phi\right\|_{L^2}\left\{\epsilon\left\||\langle D\rangle^{2\rho}\psi\right\|_{L^2}+C_\epsilon\|\psi\|_{L^2}\right\}\lesssim\left\|\phi\right\|_{L^2}\left\|\langle D\rangle^{2\rho}\psi\right\|_{L^2}
\end{gather}
and
\begin{equation}
\left|\left(A_\rho\phi,V_{\rm sing}\psi\right)_{L^2}-\left(V_{\rm sing}\phi,A_\rho\psi\right)_{L^2}\right|\lesssim\left\|\langle D\rangle^{2\rho}\phi\right\|_{L^2}\left\|\langle D\rangle^{2\rho}\psi\right\|_{L^2}
\end{equation}
for $\phi,\psi\in H^{2\rho}(\mathbb{R}^n)\cap\dom A_\rho$. By the Riesz representation theorem (\cite[Theorem II.4]{ReSi}) and Lemma \ref{lem3}, there exists a bounded operator $L_{V_{\rm sing}A_\rho}: H^{2\rho}(\mathbb{R}^n)\to\mathscr{H}_{-2\rho}\simeq H^{2\rho}(\mathbb{R}^n)^*$ such that
\begin{equation}
\left(A_\rho\phi,V_{\rm sing}\psi\right)_{L^2}-\left(V_{\rm sing}\phi,A_\rho\psi\right)_{L^2}=\left(\langle D\rangle^{-2\rho}L_{V_{\rm sing}A_\rho}\phi,\langle D\rangle^{2\rho}\psi\right)_{L^2}.
\end{equation}
We note that $\mathscr{H}_{-2\rho}$ is the completion of
\begin{equation}
\left\{\phi\in L^2(\mathbb{R}^n)\biggm| \int_{\mathbb{R}^n}\langle \xi\rangle^{-4\rho}\left|\mathscr{F}\phi(\xi)\right|^2d\xi<\infty\right\}
\end{equation}
that is regarded as the dual space of $H^{2\rho}(\mathbb{R}^n)$, and that the relation $H^{2\rho}(\mathbb{R}^n)\subset L^2(\mathbb{R}^n)\subset\mathscr{H}_{-2\rho}$ holds. We denote $L_{V_{\rm sing}A_\rho}=[V_{\rm sing},A_\rho]_{-2\rho}$ and (see also \cite[Lemma 6.2]{Iso} or the paragraphs below of \cite[Theorem 6.2.10]{AmBoGe}). Similarly, we define $[V_{\rm short},A_\rho]_{-2\rho}$ by the estimate
\begin{equation}
\left|\left(A_\rho\phi,V_{\rm short}\psi\right)_{L^2}\right|=\left|\left(\langle x\rangle^{-1}A_\rho\phi,\langle x\rangle V_{\rm short}\psi\right)_{L^2}\right|\lesssim\left\|\phi\right\|_{L^2}\left\|\psi\right\|_{L^2}.
\end{equation}
In contrast with $V_{\rm sing}$ and $V_{\rm short}$, $V_{\rm long}$ is differentiable. When $\rho<1$, the commutator $[V_{\rm long},A_\rho]$ on $\dom\langle x\rangle$ is extended to a compact operator on $L^2(\mathbb{R}^n)$ by the computations below (see \eqref{the4_7}, \eqref{the4_8}, and \eqref{the4_9}). Therefore, by the extensions of the commutators,
\begin{gather}
\left[H_\rho,A_\rho\right]_{-2\rho}=\left[\Psi_\rho\left(|D|^2\right),A_\rho\right]_{-2\rho}+\left[V_{\rm sing}+V_{\rm short},A_\rho\right]_{-2\rho}+\left[V_{\rm long},A_\rho\right]_{-2\rho}\nonumber\\
=\left[\Psi_\rho\left(|D|^2\right),A_\rho\right]^0+\left[V_{\rm sing}+V_{\rm short},A_\rho\right]_{-2\rho}+\left[V_{\rm long},A_\rho\right]\label{the4_1}
\end{gather}
holds on $H^{2\rho}(\mathbb{R}^n)$ because $H^{2\rho}(\mathbb{R}^n)\cap\dom A_\rho$ is core for $H_\rho$, and the left-hand side of \eqref{mourre_inequality1} is defined as the bounded operator on $L^2(\mathbb{R}^n)$. We note that
\begin{gather}
{\rm i}\left[\Psi_\rho\left(|D|^2\right),A_\rho\right]^0\nonumber\\
=2G_\rho\left(|D|^2\right)\left\{G_\rho\left(|D|^2\right)\langle D\rangle^2-\rho\right\}+2G_\rho\left(|D|^2\right)\left\{\rho-G_\rho\left(|D|^2\right)\right\}
\end{gather}
and that
\begin{gather}
g\left(\Psi_\rho\left(|D|^2\right)\right)G_\rho\left(|D|^2\right)\left\{\rho-G_\rho\left(|D|^2\right)\right\}g\left(\Psi_\rho\left(|D|^2\right)\right)\nonumber\\
=\rho g\left(\Psi_\rho\left(|D|^2\right)\right)\Psi'_\rho\left(|D|^2\right)\left\{1-\langle D\rangle^{2\rho-2}\right\}g\left(\Psi_\rho\left(|D|^2\right)\right)\geqslant0.
\end{gather}
We therefore have the inequality
\begin{gather}
g\left(\Psi_\rho\left(|D|^2\right)\right){\rm i}\left[\Psi_\rho\left(|D|^2\right),A_\rho\right]^0g\left(\Psi_\rho\left(|D|^2\right)\right)\nonumber\\
\geqslant2\rho g\left(\Psi_\rho\left(|D|^2\right)\right)\Psi'_\rho\left(|D|^2\right)\Psi_\rho\left(|D|^2\right)g\left(\Psi_\rho\left(|D|^2\right)\right)\label{the4_2}
\end{gather}
holds. Because
\begin{equation}
g(H_\rho)-g\left(\Psi_\rho\left(|D|^2\right)\right)=\frac{1}{2\pi{\rm i}}\int_\mathbb{C}\bar{\partial_z}\tilde{g}(z)\left(z-H_\rho\right)^{-1}V\left\{z-\Psi_\rho\left(|D|^2\right)\right\}^{-1}{\rm d}z\wedge{\rm d}\bar{z}\label{the4_3}
\end{equation}
is compact, there exists a compact operators $\hat{K}$ such that
\begin{gather}
g(H_\rho){\rm i}\left[\Psi_\rho\left(|D|^2\right),A_\rho\right]^0g(H_\rho)\nonumber\\
\geqslant2\rho g\left(\Psi_\rho\left(|D|^2\right)\right)\Psi'_\rho\left(|D|^2\right)\Psi_\rho\left(|D|^2\right)g\left(\Psi_\rho\left(|D|^2\right)\right)+\hat{K}\nonumber\\
\geqslant\frac{2\rho^2\lambda_1}{(1+\lambda_2)^{(1-\rho)/\rho}}g\left(\Psi_\rho\left(|D|^2\right)\right)^2+\hat{K}.\label{the4_4}
\end{gather}
On the right-hand side of \eqref{the4_4}, we used the relation $\Psi'_\rho=\rho/(1+\Psi_\rho)^{(1-\rho)/\rho}$ and the inequality
\begin{equation}
\int_{\lambda_1}^{\lambda_2}g(\lambda)\frac{\lambda}{(1+\lambda)^{(1-\rho)/\rho}}g(\lambda)E_{\Psi_\rho(|D|^2)}({\rm d}\lambda)\geqslant\frac{\lambda_1}{(1+\lambda_2)^{(1-\rho)/\rho}}g\left(\Psi_\rho\left(|D|^2\right)\right)^2,\label{the4_5}
\end{equation}
where $E_{\Psi_\rho(|D|^2)}$ is the spectral measure of $\Psi_\rho(|D|^2)$. Writing such that
\begin{equation}
g(H_\rho)\left[V_{\rm sing}+V_{\rm short},A_\rho\right]_{-2\rho}g(H_\rho)=g(H_\rho)\left(V_{\rm sing}+V_{\rm short}\right)\langle x\rangle\langle x\rangle^{-1}A_\rho g(H_\rho)-{\rm hc},\label{the4_6}
\end{equation}
we find that \eqref{the4_6} is compact because $\langle x\rangle(V_{\rm sing}+V_{\rm short})\langle\Psi_\rho(|D|^2)\rangle^{-1}$ is compact by Proposition \ref{prop1} and \eqref{short_range_decay}. We also write
\begin{equation}
\left[V_{\rm long}, G_\rho\left(|D|^2\right)D\cdot x\right]=\left[V_{\rm long},G_\rho\left(|D|^2\right)\right]D\cdot x+{\rm i}G_\rho\left(|D|^2\right)\nabla V_{\rm long}\cdot x\label{the4_7}
\end{equation}
on $C_0^\infty(\mathbb{R}^n)$ that is a core of $\dom\langle x\rangle$. We know that $G_\rho(|D|^2)\nabla V_{\rm long}\cdot x$ is compact by \eqref{long_range_decay}. The compactness of the commutator $\left[V_{\rm long}, A_\rho\right]$ is obtained as follows. Noting that $(z-|D|^2)^{-1}[|D|^2,V_{\rm long}](z-|D|^2)^{-1}$ is a compact operator, we compute
\begin{gather}
\left(z-|D|^2\right)^{-1}{\rm i}\left[|D|^2,V_{\rm long}\right]\left(z-|D|^2\right)^{-1}D\cdot x\nonumber\\
=\left(z-|D|^2\right)^{-1}\left(D\cdot\nabla V_{\rm long}+\nabla V_{\rm long}\cdot D\right)D\cdot x\left(z-|D|^2\right)^{-1}\nonumber\\
+\ 2\left(z-|D|^2\right)^{-1}\left[|D|^2,V_{\rm long}\right]\left(z-|D|^2\right)^{-2}|D|^2\label{the4_8}
\end{gather}
and estimate such that
\begin{equation}
\left\|\left(z-|D|^2\right)^{-1}\left[|D|^2,V_{\rm long}\right]\left(z-|D|^2\right)^{-1}D\cdot x\right\|\lesssim|{\rm Im}z|^{-2}\langle z\rangle+|{\rm Im}z|^{-3}\langle z\rangle^{3/2}.\label{the4_9}
\end{equation}
By the Helffer--Sj\"ostrand formula, we find that $\left[V_{\rm long}, A_\rho\right]$ is compact because $x\cdot\nabla V_{\rm long}\langle D\rangle^{-1}$ is also compact. From \eqref{the4_1} and \eqref{the4_4}, we have \eqref{mourre_inequality1} with a compact operator
\begin{equation}
K=\hat{K}+\frac{2\lambda_1}{(1+\lambda_2)^{(1-\rho)/\rho}}\left\{g\left(\Psi_\rho\left(|D|^2\right)\right)^2-g(H_\rho)^2\right\}+g(H_\rho){\rm i}\left[V,A_\rho\right]_{-2\rho}g(H_\rho).
\end{equation}
The case of $\rho=1$ is the traditional result given by \cite{Mo}. Because ${\rm i}[|D|^2,A]=2|D|^2=2H_1-2V$ is obtained directly, we do not have to compute \eqref{the4_3}, \eqref{the4_4}, and \eqref{the4_5}. We only note that, although $\left[V_{\rm long}, A\right]={\rm i}x\cdot\nabla V_{\rm long}$ is not compact but just bounded, $x\cdot\nabla V_{\rm long}g(H_1)$ is compact.
\end{proof}

The Mourre inequality \eqref{mourre_inequality1} provides us detailed information on the eigenvalues of $H_\rho$ as in Corollary \ref{cor1} below. To investigate the singular continuous spectrum of $H_\rho$, we have to prove the limiting absorption principle in Mourre theory. Many studies have investigated this topic, even for the $N$-body Schr\"odinger operator case (e.g., \cite{PeSiSi}, \cite{AmBoGe} and \cite{Ta}).
 
\begin{Cor}\label{cor1}
Any point in $\sigma_{\rm pp}(H_\rho)\setminus\{0\}$ is isolated and its multiplicity is at most finite, and the only accumulation point of $\sigma_{\rm pp}(H_\rho)$ can be at zero.
\end{Cor}

\begin{proof}
We already know that
\begin{equation}
\left|\left(A_\rho\phi,H_\rho\phi\right)_{L^2}-\left(H_\rho\phi,A_\rho\phi\right)_{L^2}\right|\lesssim\left\|\langle D\rangle^{2\rho}\phi\right\|_{L^2}^2\lesssim\left\|\langle H_\rho\rangle\phi\right\|_{L^2}^2\label{cor1_1}
\end{equation}
holds for $\phi\in H^{2\rho}(\mathbb{R}^n)\cap\dom A_\rho$ by the proof of Theorem \ref{the4}. Lemma \ref{lem3} and \eqref{cor1_1} imply that $H_\rho$ belongs to the class $C^1(A_\rho)$. The Mourre inequality \eqref{mourre_inequality1} and virial theorem compete our proof (see \cite[Theorem 6.2.10, Proposition 7.2.10, and Corollary 7.2.11]{AmBoGe}).
\end{proof}

\begin{Rem}\label{rem5}
If $1/2\leqslant\rho\leqslant1$, we can choose \eqref{conjugate_dilation} as the conjugate operator by virtue of {\rm \cite[Theorems 3.4 and 3.10]{IsLoSa}}. In more details, the commutator $\Psi_\rho(|D|^2)$ and $A$ on $C_0^\infty(\mathbb{R}^n)$ can be extended to a self-adjoint operator
\begin{equation}
{\rm i}\left[\Psi_\rho\left(|D|^2\right),A\right]^0=2\Psi'_\rho\left(|D|^2\right)|D|^2.
\end{equation}
Noting that Lemma {\rm \ref{lem3}} holds for replacing $A_\rho$ with $A$ and
\begin{equation}
\left|\left(A\phi,V_{\rm sing}\psi\right)_{L^2}\right|\lesssim\left\|\langle D\rangle\phi\right\|_{L^2}\left\|\langle D\rangle^{2\rho}\psi\right\|_{L^2}\label{rem5_1}
\end{equation}
also holds for $\phi,\psi\in H^{2\rho}(\mathbb{R}^n)\cap\dom A$, the commutator $[H_\rho,A]_{-2\rho}$ is defined as a bounded operator from $H^{2\rho}(\mathbb{R}^n)$ to $\mathscr{H}_{-2\rho}$. The shape of the Mourre estimate in this case is
\begin{gather}
g(H_\rho){\rm i}\left[H_\rho,A\right]_{-2\rho}g(H_\rho)\nonumber\\
\geqslant2\rho\lambda_1g(H_\rho)^2+g(H_\rho)\left\{{\rm i}\left[V_{\rm sing}+V_{\rm short},A\right]_{-2\rho}+x\cdot\nabla V_{\rm long}-2\rho V\right\}g(H_\rho)
\end{gather}
and the second term of the right-hand side is compact. If $0<\rho<1/2$, the commutator $H_\rho$ and $A_\rho$ can be extended to the map $H^1(\mathbb{R}^n)$ to $\mathscr{H}_{-1}$ by \eqref{rem5_1}. However, unfortunately, the left-hand side of the Mourre inequality can not be defined because $H^1(\mathbb{R}^n)\subsetneq H^{2\rho}(\mathbb{R}^n)$. Meanwhile, if $V$ has the long-range part only {\rm (i.e., $V=V_{\rm long}$)}, we can employ $A$ for all $0<\rho\leqslant1$ with the Mourre estimate
\begin{equation}
g(H_\rho){\rm i}\left[H_\rho,A\right]^0g(H_\rho)\geqslant2\rho\lambda_1g(H_\rho)^2+g(H_\rho)\left({\rm i}x\cdot\nabla V_{\rm long}-2\rho V\right)g(H_\rho)
\end{equation}
by {\rm\cite[Theorem 3.10]{IsLoSa}}.
\end{Rem}

We have everything arranged to prove the minimal velocity bound.
\begin{proof}[Proof of Theorem \ref{the1}]
As in the proofs before, we first assume that $\rho<1$. Let $g\in C_0^\infty((0,\infty))$ satisfy $fg=f$. Let $\chi$ and $\chi_1$ that belong to $C_0^\infty(\mathbb{R})$ satisfy that $\chi(s)=1$ if $|s|<\theta_0$ and $\chi(s)=0$ if $|s|>2\theta_0$, and that $\chi_1(s)=1$ if $|s|<2\theta_0$ and $\chi_1(s)=0$ if $|s|>3\theta_0$. The size of $\theta_0$ is to be determined later. According to \cite[Proposition 4.4.7]{DeGe}, and \cite[Theorem 2.38]{Iso}, we define the observables $\mathscr{M}(t)$ and $\mathscr{L}(t)$ by
\begin{align}
\mathscr{M}(t)&=\frac{1}{2}\left\{\Psi'_\rho\left(|D|^2\right)D-\frac{x}{2t} \right\}\cdot\frac{x}{|x|}\chi'\left(\frac{|x|}{2t}\right)+{\rm hc}+\chi\left(\frac{|x|}{2t}\right),\\
\mathscr{L}(t)&=f(H_\rho)\mathscr{M}(t)g(H_\rho)\frac{A_\rho}{t}g(H_\rho)\mathscr{M}(t)f(H_\rho).
\end{align}
Because $g(H_\rho)\dom\langle x\rangle\subset\dom A_\rho$ holds and $A_\rho g(H_\rho)\langle x\rangle^{-1}$ is a bounded operator as we proved in Lemma \ref{lem3}, $\mathscr{L}(t)$ is well-defined. By the supporting properties $\chi=\chi_1\chi$ and $\chi'=\chi_1\chi'$, we compute
\begin{equation}
\mathscr{M}(t)=\chi_1\left(\frac{|x|}{2t}\right)\mathscr{M}(t)+B(t)\label{the1_1}
\end{equation}
with
\begin{equation}
B(t)=\frac{1}{2}\sum_{j=1}^n\left[\Psi'_\rho\left(|D|^2\right)D_j,\chi_1\left(\frac{|x|}{2t}\right)\right]\frac{x_j}{|x|}\chi'\left(\frac{|x|}{2t}\right).
\end{equation}
We already know $B(t)=\mathcal{O}(t^{-1})$ from the computation \eqref{the2_4}. Moreover, by
\begin{gather}
\left\|\frac{x_k}{t}\left(z-|D|^2\right)^{-1}D_j\left[|D|^2,\chi_1\left(\frac{|x|}{2t}\right)\right]\left(z-|D|^2\right)^{-1}\right\|\nonumber\\
\lesssim t^{-1}|{\rm Im}z|^{-2}\langle z\rangle+t^{-2}|{\rm Im}z|^{-3}\langle z\rangle^{3/2}
\end{gather}
and the Helffer--Sj\"ostrand formula, we find that
\begin{equation}
\frac{x_k}{t}D_j\left[\Psi'_\rho\left(|D|^2\right),\chi_1\left(\frac{|x|}{2t}\right)\right]=\mathcal{O}\left(t^{-1}\right)\label{the1_2}
\end{equation}
for $1\leqslant j,k\leqslant n$. \eqref{the1_2} and $(x_k/t)[D_j,\chi_1(|x|/(2t))]\Psi'_\rho(|D|^2)=\mathcal{O}(t^{-1})$ yield
\begin{equation}
\frac{x_k}{t}\left[\Psi'_\rho\left(|D|^2\right)D_j,\chi_1\left(\frac{|x|}{2t}\right)\right]=\mathcal{O}\left(t^{-1}\right)\label{the1_3}
\end{equation}
and
\begin{equation}
\frac{x_k}{t}B(t)=\mathcal{O}\left(t^{-1}\right).\label{the1_4}
\end{equation}
We therefore have
\begin{equation}
\frac{A_\rho}{t}g(H_\rho)\left\langle\frac{x}{2t}\right\rangle^{-1}=\mathcal{O}(1)\label{the1_5}
\end{equation}
and $(A_\rho/t)g(H_\rho)\mathscr{M}(t)=\mathcal{O}(1)$ from \eqref{the1_1} and \eqref{the1_4}. This implies that $\mathscr{L}(t)=\mathcal{O}(1)$. We write
\begin{gather}
\mathbb{D}_{H_\rho}\mathscr{L}(t)=I_1(t)+I_2(t)+I_3(t)+I_4(t),
\end{gather}
where
\begin{align}
I_1(t)&=f(H_\rho)\left\{\mathbb{D}_{\Psi_\rho(|D|^2)}\mathscr{M}(t)\right\}g(H_\rho)\frac{A_\rho}{t}g(H_\rho)\mathscr{M}(t)f(H_\rho)+{\rm hc},\\
I_2(t)&=f(H_\rho){\rm i}\left[V,\mathscr{M}(t)\right]g(H_\rho)\frac{A_\rho}{t}g(H_\rho)\mathscr{M}(t)f(H_\rho)+{\rm hc},\\
I_3(t)&=-\frac{1}{t}f(H_\rho)\mathscr{M}(t)g(H_\rho)\frac{A_\rho}{t}g(H_\rho)\mathscr{M}(t)f(H_\rho),\\
I_4(t)&=\frac{1}{t}f(H_\rho)\mathscr{M}(t)g(H_\rho){\rm i}\left[H_\rho,A_\rho\right]_{-2\rho}g(H_\rho)\mathscr{M}(t)f(H_\rho).
\end{align}

\noindent
{\bf Estimate for $I_3$ and $I_4$.}\quad By the same computations as \eqref{the2_7}, \eqref{the1_1}, \eqref{the1_4} and
\begin{equation}
\mathscr{M}(t)=\mathscr{M}(t)\chi_1\left(\frac{|x|}{2t}\right)+\frac{1}{2}\sum_{j=1}^n\chi'\left(\frac{|x|}{2t}\right)\frac{x_j}{|x|}\left[\chi_1\left(\frac{|x|}{2t}\right),\Psi'_\rho\left(|D|^2\right)D_j\right],\label{the1_6}
\end{equation}
we have
\begin{equation}
I_3(t)=-\frac{1}{t}f(H_\rho)\mathscr{M}(t)g(H_\rho)\chi_1\left(\frac{|x|}{2t}\right)\frac{A_\rho}{t}\chi_1\left(\frac{|x|}{2t}\right)g(H_\rho)\mathscr{M}(t)f(H_\rho)+\mathcal{O}\left(t^{-2}\right).\label{the1_7}
\end{equation}
We have, using $[G_\rho(D)D_j,\chi_1(|x|/(2t))]=\mathcal{O}(t^{-1})$ similar to \eqref{the2_4},
\begin{gather}
g(H_\rho)\chi_1\left(\frac{|x|}{2t}\right)\frac{A_\rho}{t}\chi_1\left(\frac{|x|}{2t}\right)g(H_\rho)\nonumber\\
=g(H_\rho)G_\rho\left(|D|^2\right)D\cdot\frac{x}{|x|}\chi_1\left(\frac{|x|}{2t}\right)\frac{|x|}{2t}\chi_1\left(\frac{|x|}{2t}\right)g(H_\rho)+{\rm hc}+\mathcal{O}\left(t^{-1}\right)\nonumber\\
\leqslant2\left\|g(H_\rho)G_\rho\left(|D|^2\right)D\cdot\frac{x}{|x|}\chi_1\left(\frac{|x|}{2t}\right)\right\|\left\|\frac{|x|}{2t}\chi_1\left(\frac{|x|}{2t}\right)g(H_\rho)\right\|+\mathcal{O}\left(t^{-1}\right),\label{the1_8}
\end{gather}
and we then estimate
\begin{equation}
I_3(t)\geqslant-\frac{\theta}{t}f(H_\rho)\mathscr{M}(t)^2f(H_\rho)+\mathcal{O}\left(t^{-2}\right),\label{the1_9}
\end{equation}
where we put $\theta=6\theta_0\|g(H_\rho)G_\rho(|D|^2)D\cdot x/|x|\|\|g(H_\rho)\|$.
We next estimate $I_4$. It follows from \eqref{the2_7} that
\begin{equation}
\left[g(H_\rho),\mathscr{M}(t)\right]=\frac{1}{2}\left[g(H_\rho),\Psi'_\rho\left(|D|^2\right)D\cdot\frac{x}{|x|}\chi'\left(\frac{|x|}{2t}\right)+{\rm hc}\right]+\mathcal{O}\left(t^{-1}\right).\label{the1_10}
\end{equation}
By \eqref{commutator_left}, we compute
\begin{gather}
\left[\Psi_\rho\left(|D|^2\right),\Psi'_\rho\left(|D|^2\right)D_j\frac{x_j}{|x|}\chi'\left(\frac{|x|}{2t}\right)\right]\nonumber\\
=\Psi'_\rho\left(|D|^2\right)D_j\left[|D|^2,\frac{x_j}{|x|}\chi'\left(\frac{|x|}{2t}\right)\right]\Psi'_\rho\left(|D|^2\right)+\Psi'_\rho\left(|D|^2\right)D_j\mathcal{O}\left(t^{-2}\right)\label{the1_11}
\end{gather}
and
\begin{gather}
\left\|\left(z-H_\rho\right)^{-1}\left[\Psi_\rho\left(|D|^2\right),\Psi'_\rho\left(|D|^2\right)D\cdot\frac{x}{|x|}\chi'\left(\frac{|x|}{2t}\right)\right]\left(z-H_\rho\right)^{-1}\right\|\nonumber\\
\lesssim t^{-1}|{\rm Im}z|^{-2}\langle z\rangle^2,\label{the1_12}
\end{gather}
recalling \eqref{the2_5} and \eqref{the2_6}. We write the commutator such that
\begin{gather}
\left[V_{\rm sing},\Psi'_\rho\left(|D|^2\right)D_j\frac{x_j}{|x|}\chi'\left(\frac{|x|}{2t}\right)\right]=V_{\rm sing}\langle x\rangle^{\gamma_{\rm sing}}\mathcal{O}\left(t^{-\gamma_{\rm sing}}\right)\Psi'_\rho\left(|D|^2\right)D_j\nonumber\\
+\ V_{\rm sing}\left[\Psi'_\rho\left(|D|^2\right)D_j,\frac{x_j}{|x|}\chi'\left(\frac{|x|}{2t}\right)\right]-\Psi'_\rho\left(|D|^2\right)D_j\mathcal{O}\left(t^{-\gamma_{\rm sing}}\right)\langle x\rangle^{\gamma_{\rm sing}}V_{\rm sing}.\label{the1_13}
\end{gather}
We here used $\langle x\rangle^{-\gamma_{\rm sing}}\chi'(|x|/(2t))=\mathcal{O}(t^{-\gamma_{\rm sing}})$. By the same computations as \eqref{the3_16}, we have
\begin{equation}
\langle x\rangle^{-\gamma_{\rm sing}}\left[\Psi'_\rho\left(|D|^2\right)D_j,\frac{x_j}{|x|}\chi'\left(\frac{|x|}{2t}\right)\right]=\mathcal{O}\left(t^{-2}\right).\label{the1_14}
\end{equation}
From \eqref{the1_13} and \eqref{the1_14}, we estimate
\begin{gather}
\left\|\left(z-H_\rho\right)^{-1}\left[V_{\rm sing},\Psi'_\rho\left(|D|^2\right)D\cdot\frac{x}{|x|}\chi'\left(\frac{|x|}{2t}\right)\right]\left(z-H_\rho\right)^{-1}\right\|\nonumber\\
\lesssim t^{-\gamma_{\rm sing}}|{\rm Im}z|^{-2}\langle z\rangle^2+t^{-2}|{\rm Im}z|^{-2}\langle z\rangle,\label{the1_15}
\end{gather}
using \eqref{the2_6} and $\|(z-H_\rho)^{-1}V_{\rm sing}\langle x\rangle^{\gamma_{\rm sing}}\|\lesssim|{\rm Im}z|^{-1}\langle z\rangle$. Because $V_{\rm short}\langle x\rangle^{\gamma_{\rm short}}$ is bounded by \eqref{short_range_decay}, we also estimate
\begin{gather}
\left\|\left(z-H_\rho\right)^{-1}\left[V_{\rm short},\Psi'_\rho\left(|D|^2\right)D\cdot\frac{x}{|x|}\chi'\left(\frac{|x|}{2t}\right)\right]\left(z-H_\rho\right)^{-1}\right\|\nonumber\\
\lesssim t^{-\gamma_{\rm short}}|{\rm Im}z|^{-2}\langle z\rangle+t^{-2}|{\rm Im}z|^{-2}.\label{the1_16}
\end{gather}
Noting that an almost analytic extension of $g$ has compact support, from \eqref{the3_12}, \eqref{the1_10}, \eqref{the1_12}, \eqref{the1_15}, and \eqref{the1_16}, we have
\begin{equation}
\left[g(H_\rho),\mathscr{M}(t)\right]=\mathcal{O}\left(t^{-\min\{\gamma_{\rm sing},\gamma_{\rm short},1+\gamma_{\rm long},2\}}\right)+\mathcal{O}\left(t^{-1}\right)=\mathcal{O}\left(t^{-1}\right).\label{the1_17}
\end{equation}
Incidentally, let $\lambda_1$ and $\lambda_2$ in Theorem \ref{the4} satisfy $(\lambda_1,\lambda_2)\cap\sigma_{\rm pp}(H_\rho)=\emptyset$. For $\lambda_1<\lambda<\lambda_2$, we take $0<\delta<\min\left\{\lambda-\lambda_1,\lambda_2-\lambda\right\}$. $\lambda\not\in\sigma_{\rm pp}(H_\rho)$ is equivalent to the point spectral measure $E_{H_\rho}(\{\lambda\})$ being zero. This implies that $E_{H_\rho}((\lambda-\delta,\lambda+\delta))\rightarrow0$ as $\delta\rightarrow0$ in the strong norm sense of $L^2(\mathbb{R}^n)$ and that, for the compact operator $K$ of \eqref{mourre_inequality1}, $E_{H_\rho}((\lambda-\delta,\lambda+\delta))K\rightarrow0$ as $\delta\rightarrow0$ in operator norm sense of $L^2(\mathbb{R}^n)$. Therefore, Theorem \ref{the4} yields
\begin{gather}
E_{H_\rho}\left((\lambda-\delta,\lambda+\delta)\right){\rm i}\left[H_\rho,A_\rho\right]_{-2\rho}E_{H_\rho}\left((\lambda-\delta,\lambda+\delta)\right)\nonumber\\
\geqslant\frac{\rho^2\lambda_1}{(1+\lambda_2)^{(1-\rho)/\rho}}E_{H_\rho}\left((\lambda-\delta,\lambda+\delta)\right)\label{mourre_inequality2}
\end{gather}
for a small $\delta>0$. We assume that $\supp g$ is sufficiently small without loss of generality because, if not, $\supp g$ can be covered by $\bigcup_{k=1}^N\supp g_k$ where $\supp g_k$ is small (see the proof of \cite[Proposition 4.4.7]{DeGe}). By virtue of \eqref{the1_17} and \eqref{mourre_inequality2}, there exists $c=c_{\rho g}>0$ such that $I_4$ is estimated as
\begin{equation}
I_4(t)\geqslant\frac{c}{t}f(H_\rho)\mathscr{M}(t)g(H_\rho)^2\mathscr{M}(t)f(H_\rho)=\frac{c}{t}f(H_\rho)\mathscr{M}(t)^2f(H_\rho)+\mathcal{O}\left(t^{-2}\right).\label{the1_18}
\end{equation}
We here choose $\theta_0>0$ which satisfies $0<\theta<c$ noting the definition of $\theta$, and put $\mathscr{K}(t)$
\begin{equation}
\mathscr{K}(t)=\frac{1}{2}\left\{\Psi'_\rho\left(|D|^2\right)D-\frac{x}{2t} \right\}\cdot\frac{x}{|x|}\chi'\left(\frac{|x|}{2t}\right)+{\rm hc}
\end{equation}
as in \eqref{the3_13}. From \eqref{the1_9} and \eqref{the1_18}, using the inequality
\begin{equation}
\mathscr{M}(t)^2\geqslant\chi\left(\frac{|x|}{2t}\right)^2+\mathscr{K}(t)^2-\left\{2\mathscr{K}(t)^2+\frac{1}{2}\chi\left(\frac{|x|}{2t}\right)^2\right\}=\frac{1}{2}\chi\left(\frac{|x|}{2t}\right)^2-\mathscr{K}(t)^2,
\end{equation}
we have
\begin{gather}
I_3(t)+I_4(t)\geqslant\frac{c-\theta}{2t}f(H_\rho)\chi\left(\frac{|x|}{2t}\right)^2f(H_\rho)-\frac{c-\theta}{t}f(H_\rho)\mathscr{K}(t)^2f(H_\rho)+\mathcal{O}\left(t^{-2}\right).\label{the1_19}
\end{gather}
We note that, by virtue of Theorem \ref{the3},
\begin{gather}
\int_1^\infty\left|\left(\mathscr{K}(t)^2f(H_\rho)e^{-{\rm i}tH_\rho}\phi,f(H_\rho)e^{-{\rm i}tH_\rho}\phi\right)_{L^2}\right|\frac{{\rm d}t}{t}\lesssim\|\phi\|_{L^2}^2\label{the1_20}
\end{gather}
holds because $f(H_\rho)\mathscr{K}(t)^2f(H_\rho)$ has the following shape
\begin{gather}
f(H_\rho)\mathscr{K}(t)^2f(H_\rho)=f(H_\rho)\left\{\Psi'_\rho\left(|D|^2\right)D-\frac{x}{2t}\right\}\cdot\frac{x}{|x|}\chi'\left(\frac{|x|}{2t}\right)\nonumber\\
\times\ \chi'\left(\frac{|x|}{2t}\right)\frac{x}{|x|}\cdot\left\{\Psi'_\rho\left(|D|^2\right)D-\frac{x}{2t}\right\}f(H_\rho)+\mathcal{O}\left(t^{-1}\right),
\end{gather}
by \eqref{the2_4}.

\noindent
{\bf Estimate for $I_2$.}\quad By \eqref{the3_12}, \eqref{the1_15}, and \eqref{the1_16}, replacing $(z-H_\rho)^{-1}$ by $\langle H_\rho\rangle^{-1}$ in \eqref{the1_15} and \eqref{the1_16}, we have
\begin{equation}
\langle H_\rho\rangle^{-1}{\rm i}\left[V,\mathscr{M}(t)\right]\langle H_\rho\rangle^{-1}=\mathcal{O}\left(t^{-\min\{\gamma_{\rm sing},\gamma_{\rm short},1+\gamma_{\rm long},2\}}\right)
\end{equation}
and
\begin{equation}
I_2(t)=\mathcal{O}\left(t^{-\min\{\gamma_{\rm sing},\gamma_{\rm short},1+\gamma_{\rm long},2\}}\right).\label{the1_21}
\end{equation}
\\
\noindent
{\bf Estimate for $I_1$.}\quad Put $R(x)=\chi(|x|)$. Then, by the formula \eqref{heisenberg2}, $I_1$ is
\begin{gather}
I_1(t)=\frac{1}{t}f(H_\rho)\left\{\Psi'_\rho\left(|D|^2\right)D-\frac{x}{2t}\right\}\cdot\left(\nabla^2R\right)\left(\frac{x}{2t}\right)\left\{\Psi'_\rho\left(|D|^2\right)D-\frac{x}{2t}\right\}\nonumber\\
\times\ g(H_\rho)\frac{A_\rho}{t}g(H_\rho)\mathscr{M}(t)f(H_\rho)+{\rm hc}+\mathcal{O}\left(t^{-2}\right)=I_5(t)+I_6(t)+\mathcal{O}(t^{-2}),\label{the1_22}
\end{gather}
where we defined $I_5$ and $I_6$ by
\begin{gather}
I_5(t)=\frac{1}{t}f(H_\rho)\left\{\Psi'_\rho\left(|D|^2\right)D-\frac{x}{2t}\right\}\cdot\left(\nabla^2R\right)\left(\frac{x}{2t}\right)\left\{\Psi'_\rho\left(|D|^2\right)D-\frac{x}{2t}\right\}\nonumber\\
\quad\times\ g(H_\rho)\frac{A_\rho}{t}g(H_\rho)\chi'\left(\frac{|x|}{2t}\right)\frac{x}{|x|}\cdot\left\{\Psi'_\rho\left(|D|^2\right)D-\frac{x}{2t}\right\}f(H_\rho)+{\rm hc},\\
I_6(t)=\frac{1}{t}f(H_\rho)\left\{\Psi'_\rho\left(|D|^2\right)D-\frac{x}{2t}\right\}\cdot\left(\nabla^2R\right)\left(\frac{x}{2t}\right)\left\{\Psi'_\rho\left(|D|^2\right)D-\frac{x}{2t}\right\}\nonumber\\
\quad\times\ g(H_\rho)\frac{A_\rho}{t}g(H_\rho)\chi\left(\frac{|x|}{2t}\right)f(H_\rho)+{\rm hc},
\end{gather}
using \eqref{the2_4},
\begin{equation}
\mathscr{M}(t)=\chi'\left(\frac{|x|}{2t}\right)\frac{x}{|x|}\cdot\left\{\Psi'_\rho\left(|D|^2\right)D-\frac{x}{2t}\right\}+\chi\left(\frac{|x|}{2t}\right)+\mathcal{O}(t^{-1})\label{the1_23}
\end{equation}
and \eqref{the1_5}. Let $\chi_2\in C_0^\infty((\theta_0/2,\infty))$ satisfy $\chi'=\chi'\chi_2$. We write $I_5$ such that
\begin{gather}
I_5(t)=\frac{1}{t}f(H_\rho)\sum_{j=1}^n\left\{\Psi'_\rho\left(|D|^2\right)D_j-\frac{x_j}{2t}\right\}\chi_2\left(\frac{|x|}{2t}\right)\nonumber\\
\times\ \mathcal{O}\left(1\right)\sum_{k=1}^n\chi_2\left(\frac{|x|}{2t}\right)\left\{\Psi'_\rho\left(|D|^2\right)D_k-\frac{x_k}{2t}\right\}f(H_\rho)+\mathcal{O}\left(t^{-2}\right),\label{the1_24}
\end{gather}
where we also used \eqref{the1_5}. We finally estimate $I_6$. By the same computations as \eqref{the1_11}, \eqref{the1_12}, \eqref{the1_13}, \eqref{the1_14} and \eqref{the1_15}, we have
\begin{equation}
\left[\chi_2\left(\frac{|x|}{2t}\right)\Psi'_\rho\left(|D|^2\right)D_j,g(H_\rho)\right]=\mathcal{O}\left(t^{-\min\{\gamma_{\rm sing},\gamma_{\rm short},1+\gamma_{\rm long},2\}}\right).\label{the1_25}
\end{equation}
We also have
\begin{equation}
\left[\chi_2\left(\frac{|x|}{2t}\right)\frac{x_j}{t},g(H_\rho)\right]=\mathcal{O}(t^{-1})\label{the1_26}
\end{equation}
by \eqref{the2_7}. \eqref{the1_25} and \eqref{the1_26} imply
\begin{equation}
\left[\chi_2\left(\frac{|x|}{2t}\right)\left\{\Psi'_\rho\left(|D|^2\right)D_j-\frac{x_j}{2t}\right\},g(H_\rho)\right]=\mathcal{O}(t^{-1}).\label{the1_27}
\end{equation}
We note that
\begin{equation}
\left[\chi_2\left(\frac{|x|}{2t}\right)\frac{x_j}{t},\frac{x_k}{t}\Psi'_\rho\left(|D|^2\right)D_k\right]=\frac{x_k}{t}\left[\chi_2\left(\frac{|x|}{2t}\right)\frac{x_j}{t},\Psi'_\rho\left(|D|^2\right)D_k\right]=\mathcal{O}\left(t^{-1}\right)\label{the1_28}
\end{equation}
by \eqref{the1_3} and that
\begin{gather}
\left[\chi_2\left(\frac{|x|}{2t}\right)\Psi'_\rho\left(|D|^2\right)D_j,\frac{x_k}{t}\Psi'_\rho\left(|D|^2\right)D_k\right]\nonumber\\
=\chi_2\left(\frac{|x|}{2t}\right)\mathcal{O}(t^{-1})\Psi'_\rho\left(|D|^2\right)D_k+\frac{x_k}{t}\mathcal{O}(t^{-1})\Psi'_\rho\left(|D|^2\right)D_j\label{the1_29}
\end{gather}
for $1\leqslant j,k\leqslant n$, where we used \eqref{the1_3} again in the second term on the right-hand side of \eqref{the1_29}.
From \eqref{the1_28} and \eqref{the1_29}, we have
\begin{equation}
\left[\chi_2\left(\frac{|x|}{2t}\right)\left\{\Psi'_\rho\left(|D|^2\right)D_j-\frac{x_j}{2t}\right\},\frac{A_\rho}{t}\right]g(H_\rho)=\mathcal{O}\left(t^{-1}\right).\label{the1_30}
\end{equation}
Clearly,
\begin{equation}
\frac{A_\rho}{t}g(H_\rho)\left[\chi_2\left(\frac{|x|}{2t}\right)\left\{\Psi'_\rho\left(|D|^2\right)D_j-\frac{x_j}{2t}\right\},\chi\left(\frac{|x|}{2t}\right)\right]=\mathcal{O}(t^{-1})\label{the1_31}
\end{equation}
holds by \eqref{the2_4} and \eqref{the1_5}. Combining \eqref{the1_27}, \eqref{the1_30} and \eqref{the1_31}, we have
\begin{equation}
\left\langle\frac{x}{2t}\right\rangle^{-1}\left[\chi_2\left(\frac{|x|}{2t}\right)\left\{\Psi'_\rho\left(|D|^2\right)D_j-\frac{x_j}{2t}\right\},g(H_\rho)\frac{A_\rho}{t}g(H_\rho)\chi\left(\frac{|x|}{2t}\right)\right]=\mathcal{O}\left(t^{-1}\right).\label{the1_32}
\end{equation}
By \eqref{the1_32}, we find that $I_6$ has the estimate
\begin{gather}
I_6(t)=\frac{1}{t}f(H_\rho)\sum_{j=1}^n\left\{\Psi'_\rho\left(|D|^2\right)D_j-\frac{x_j}{2t}\right\}\chi_2\left(\frac{|x|}{2t}\right)\nonumber\\
\times\ \mathcal{O}\left(1\right)\sum_{k=1}^n\chi_2\left(\frac{|x|}{2t}\right)\left\{\Psi'_\rho\left(|D|^2\right)D_k-\frac{x_k}{2t}\right\}f(H_\rho)+\mathcal{O}\left(t^{-2}\right).\label{the1_33}
\end{gather}
By virtue of Theorem \ref{the3}, \eqref{the1_22}, \eqref{the1_24} and \eqref{the1_33},
\begin{gather}
\int_1^\infty\left|\left(I_1(t)e^{-{\rm i}tH_\rho}\phi,e^{-{\rm i}tH_\rho}\phi\right)_{L^2}\right|{\rm d}t\lesssim\|\phi\|_{L^2}^2\label{the1_34}
\end{gather}
is obtained.

There exists a constant $C>0$ such that
\begin{gather}
\left(\left\{\mathbb{D}_{H_\rho}\mathscr{L}(t)\right\}e^{-{\rm i}tH_\rho}\phi,e^{-{\rm i}tH_\rho}\phi\right)_{L^2}\geqslant\frac{c-\theta}{2t}\left\|\chi\left(\frac{|x|}{2t}\right)f(H_\rho)e^{-{\rm i}tH_\rho}\phi\right\|_{L^2}^2\nonumber\\
-\frac{c-\theta}{t}\left|\left(\mathscr{K}(t)^2f(H_\rho)e^{-{\rm i}tH_\rho}\phi,f(H_\rho)e^{-{\rm i}tH_\rho}\phi\right)_{L^2}\right|\nonumber\\
-\left|\left(I_1(t)e^{-{\rm i}tH_\rho}\phi,e^{-{\rm i}tH_\rho}\phi\right)_{L^2}\right|-Ct^{-\min\{\gamma_{\rm sing},\gamma_{\rm short},1+\gamma_{\rm long},2\}}\|\phi\|_{L^2}^2
\end{gather}
holds by \eqref{the1_19} and \eqref{the1_21}. This completes our proof for $0<\rho<1$ by \eqref{the1_20}, \eqref{the1_34}, and $\min\{\gamma_{\rm sing},\gamma_{\rm short},1+\gamma_{\rm long},2\}>1$. In the case where $\rho=1$, as in the proofs before, we simply replace $\Psi'_\rho$ by $1$ and reduce many of the computations. For more details, see \cite[Proposition 4.4.7]{DeGe} or \cite[Theorem 2.38]{Iso}.
\end{proof}

\begin{Rem}\label{rem6}
If $1/2\leqslant\rho\leqslant1$, we can prove Theorem {\rm\ref{the1}} even by adopting the conjugate operator $A$ with some modifications to the proof presented above. In particular, $Ag(H_\rho)\langle x\rangle^{-1}$ is a bounded operator because $\langle D\rangle g(H_\rho)$ is a bounded operator.
\end{Rem}

\noindent\textbf{Acknowledgments.} This study was supported by JSPS KAKENHI Grant Numbers JP20K03625 and JP21K03279. The author thanks Serge Richard and Avraham Soffer for their valuable suggestions and comments.\\




\end{document}